\PassOptionsToPackage{unicode}{hyperref}
\PassOptionsToPackage{hyphens}{url}
\PassOptionsToPackage{dvipsnames,svgnames,x11names}{xcolor}
\documentclass[12pt]{article}

\usepackage{amsmath,amssymb,amsfonts}
\usepackage{amsthm}
\usepackage{algorithm}
\usepackage{algpseudocode}
\usepackage{array}
\usepackage{appendix}
\usepackage{booktabs}
\usepackage{caption}
\usepackage{etoolbox}
\usepackage{float}
\usepackage{graphicx}
\usepackage{hyperref}
\usepackage{iftex}
\usepackage{longtable}
\usepackage{natbib}
\usepackage{rotating}
\usepackage{subcaption}
\usepackage{xcolor}
\usepackage{xurl}

\ifPDFTeX
  \usepackage[T1]{fontenc}
  \usepackage[utf8]{inputenc}
  \usepackage{textcomp}
\else
  \usepackage{unicode-math}
  \defaultfontfeatures{Scale=MatchLowercase}
  \defaultfontfeatures[\rmfamily]{Ligatures=TeX,Scale=1}
\fi
\usepackage{lmodern}
\IfFileExists{microtype.sty}{\usepackage[]{microtype}\UseMicrotypeSet[protrusion]{basicmath}}{}
\makeatletter
\patchcmd\longtable{\par}{\if@noskipsec\mbox{}\fi\par}{}{}
\setkeys{Gin}{width=\linewidth,height=\textheight,keepaspectratio}
\def\fps@figure{htbp}
\floatstyle{ruled}
\@ifundefined{c@chapter}{\newfloat{codelisting}{h}{lop}}{\newfloat{codelisting}{h}{lop}[chapter]}
\makeatother
\floatname{codelisting}{Listing}
\hypersetup{
  pdftitle={Exact grid-free confidence-region computation from dependent p-value functions under arbitrary dependence},
  pdfauthor={},
  pdfkeywords={dependent p-value functions, exact grid-free inference, confidence-region computation, voting algorithms, arbitrary dependence, conformal ensembling},
  colorlinks=true,
  linkcolor=blue,
  citecolor=blue,
  urlcolor=blue
}
\theoremstyle{plain}
\newtheorem{theorem}{Theorem}

\newtheorem{proposition}[theorem]{Proposition}

\theoremstyle{remark}

\theoremstyle{definition}
\newtheorem{definition}{Definition}

\newcommand{\E}{\mathbb{E}}
\newcommand{\1}{\mathbf{1}}
\newcommand{\ind}[1]{\1\{#1\}}
\newcommand{\piOrd}[1]{\pi_{(#1)}}
\newcommand{\Vote}{\operatorname{Vote}}
\newcolumntype{L}[1]{>{\raggedright\arraybackslash}p{#1}}
\providecommand{\anon}{1}

\begin{document}
\def\spacingset#1{\renewcommand{\baselinestretch}{#1}\small\normalsize}
\spacingset{1}

\if1\anon
{
  \title{\bf Exact grid-free confidence-region computation from dependent $p$-value functions under arbitrary dependence}
  \author{Yaohui Lin\thanks{}\\
    School of Mathematical Sciences, South China Normal University\\
    \texttt{linyuan@m.scnu.edu.cn}}
  \date{}
  \maketitle
  \vspace{-30pt}
}
\fi

\if0\anon
{
  \bigskip
  \bigskip
  \bigskip
  \begin{center}
    {\LARGE\bf Exact grid-free confidence-region computation from dependent $p$-value functions under arbitrary dependence}
  \end{center}
  \medskip
}
\fi

\begin{abstract}
Repeated sample splitting, cross-fitting, conformal ensembling, and related randomized workflows often produce multiple dependent valid $p$-value functions for the same target. Existing $p$-merging theory guarantees pointwise validity under arbitrary dependence, but turning the merged output into a confidence region typically requires repeated evaluation on a grid over the parameter space. This inversion step can be computationally costly, approximation-dependent, and increasingly difficult to scale.

We study when confidence regions can instead be computed exactly from split-wise regions without gridding the parameter space. Our main structural result shows that mergers induced by step calibrators form a broad exactly executable class at the region level, and a partial converse within the calibrator-induced family indicates that exact executability is closely tied to step structure. Within this class, we develop exact voting algorithms, including an adaptive multi-quantile contour aggregator that avoids pre-specifying a single order-statistic threshold while preserving finite-sample validity under arbitrary dependence.

Simulation studies on repeated-split regression and conformal prediction, together with real-data regression examples, show that the proposed methods provide stable, robustness-oriented exact inference with substantial runtime gains over grid-inversion comparators. A grid-resolution benchmark shows that fixed-$k$ voting and adaptive multi-quantile voting are essentially insensitive to inversion-grid refinement, while a multidimensional single-threshold stress test illustrates the dimensional blow-up faced by grid-inversion baselines in a simple box-geometry setting. Taken together, the results show that arbitrary-dependence contour merging can be turned into an exactly executable region-computation framework rather than merely a pointwise validity device.
\end{abstract}


\vfill
\newpage
\spacingset{1.8}

\section{Introduction}\label{sec:intro}

Repeated randomization is now routine in statistical inference. Sample splitting, cross-fitting, conformal ensembling, and related workflows (see, e.g., \cite{MeinshausenMeierBuhlmann2009,WassermanRoeder2009,ChernozhukovEtAl2018,Vovk2015,BarberCandesRamdasTibshirani2021}) are often used to stabilize procedures or recover validity after data-adaptive model fitting. A common by-product is that the same inferential target is accompanied not by a single valid $p$-value function, but by a collection of dependent valid $p$-value functions generated from different random splits or randomization draws. The practical question is how to aggregate these outputs into one coherent confidence procedure.

At the pointwise level, this question is largely understood. As shown in the literature on $p$-value merging under arbitrary dependence (see \cite{VovkWang2020,VovkWangWang2022}), a scalar merging rule can be applied pointwise in the parameter to obtain a valid aggregated $p$-value function. The remaining difficulty is computational rather than purely probabilistic: pointwise validity does not by itself provide a practically usable confidence region. In continuous or high-dimensional parameter spaces, the merged function typically has to be evaluated on a dense grid and then inverted numerically, a step that can dominate runtime, introduce approximation error, and scale poorly with dimension.

\subsection{From pointwise validity to executable regions}
Pointwise lifting is conceptually simple but can be computationally impractical when $\Theta$ is continuous or high-dimensional: evaluating $\pi_{\mathrm{agg}}(\theta)$ on a fine grid can dominate the overall cost. Related conformal-prediction work has highlighted similar efficiency bottlenecks in covariate-shift, full-conformal, and non-exchangeable settings, motivating dedicated acceleration methods \cite{PrinsterSariaLiu2023,Li2024FastExactConformal,PlassierEtAl2024}. The issue also arises in ordinary statistical workflows. A practitioner may repeat sample splitting to reduce split-lottery effects, refit a model across random subsets, or query prediction and parameter regions for many target values. In these settings the split-wise regions are often already available as intervals, boxes, ellipsoids, or other sets with efficient intersection and union operations, while the merged pointwise contour still has to be inverted numerically.

A typical example arises in repeated-split conformal prediction for a continuous response. In many implementations, each random training-calibration split naturally returns a prediction interval at a requested level, and the analyst may query the resulting procedure across several test covariates, nominal levels, or random seeds to mitigate split-lottery effects. A generic pointwise \(p\)-merging implementation must then recover or evaluate the split-wise \(p\)-value functions over candidate response values before the merged contour can be inverted, with this numerical inversion repeated across the queried targets. A vote-implementable merger instead uses exactly the region-level objects already produced by the workflow: split-wise intervals at a small number of explicitly determined adjusted levels. Changing the nominal level only changes these adjusted levels, so the aggregated prediction region is obtained by exact interval operations rather than by re-evaluating a merged contour on a response grid.

This suggests a different computational strategy: instead of evaluating a merged contour point by point, can one compute the aggregated region directly from finitely many split-wise regions? For contour-based inference, including inferential-model plausibility analysis, this is not a cosmetic implementation issue. One often queries the full nested family $\{C_\alpha\}_{\alpha\in(0,1)}$, repeats the construction across many targets, or changes the nominal level during sensitivity analysis. In such cases, the inversion step becomes an avoidable computational and numerical layer. Compared with this pointwise route, direct region-level computation preserves the finite-sample validity of the merger while avoiding discretization error from the inversion grid. It also differs from asymptotic shortcuts: the construction is exact at the region level once the split-wise regions are available. This motivates the following question:

\begin{quote}
\emph{Which contour aggregation rules admit an exact region-level implementation, i.e., allow computing $C_{\mathrm{agg}}(\alpha)=\{\theta: \pi_{\mathrm{agg}}(\theta)>\alpha\}$ directly from the split-wise regions $\{C_r(\cdot)\}$ without gridding $\Theta$?}
\end{quote}

Our answer is built around a single organizing idea: exact region-level executability is a structural property of the merger, not a post-hoc implementation detail. Once that property is made explicit, a broad class of contour mergers becomes simultaneously dependence-robust, interpretable, and exactly implementable. Importantly, this exactness is not merely a computational convenience: within the calibrator-induced family, our partial converse shows that, under a mild nondegeneracy condition, fixed-level pointwise-universal executability is closely tied to step structure.

The nearby literature has three especially relevant adjacent strands. Pointwise $p$-merging gives dependence-robust contour values (see \cite{VovkWang2020,VovkWangWang2022}). Set-level majority-vote aggregation gives set merging under arbitrary dependence (see \cite{GasparinRamdas2024}). Recent SAT/data-light uncertainty-set merging work studies synthetic tests, test inversion, and admissibility for deterministic set merging (see \cite{QinHeGangXia2024}). Our object is different: we study \emph{contour-level aggregation} and ask when an entire merged \emph{nested family} can be executed exactly at the region level without gridding $\Theta$. The distinction is substantive rather than terminological. A merged set at a single nominal level does not automatically provide a coherent, exactly queryable contour family across all levels, whereas that family is the basic output of contour-based inference.

\subsection{Contributions}
This paper is organized around one claim: under arbitrary dependence, the practically decisive question is not only whether a contour merger is pointwise valid, but whether the induced nested region family is exactly executable. We make four contributions. (1) We formalize \emph{finite-layer vote-implementability} as the contour-level property that bridges merged contour values and exact region inversion. For mergers induced by step calibrators, Theorem~\ref{thm:step-cal} yields a broad exactly executable class under arbitrary dependence. (2) A fixed-level partial converse, stated as Proposition~\ref{prop:step-necessity} in Appendix~\ref{app:partial-converse}, shows that, under a mild nondegeneracy condition, pointwise-universal uniform vote-implementability inside the calibrator-induced class essentially forces step structure. This shows that the executable class is not an arbitrary algorithmic convenience but a structurally distinguished subclass of contour mergers. (3) We construct an adaptive multi-quantile (generalized Hommel-type) contour aggregator that removes the need to commit to a single order-statistic threshold while preserving exact grid-free computation through a transparent multi-threshold voting algorithm (Proposition~\ref{prop:ghom-vote} and Algorithm~\ref{alg:multivote}). (4) We show empirically that these methods occupy a useful point on the efficiency--computation frontier: they are dependence-robust, split-stable, tuning-light by default, and substantially faster than grid-inversion comparators in our benchmarks, while trading off region length against robustness and threshold-free operation relative to oracle-tuned fixed-$k$ choices. In this sense, the paper complements the existing arbitrary-dependence validity theory with a contour-level executability theory and an associated algorithmic framework.

Viewed relative to adjacent work, pointwise $p$-merging tells us how to obtain a valid contour value at each $\theta$, while set-level voting and SAT/data-light methods construct a valid set at a chosen confidence level. Our contribution sits between these views: we ask when a contour merger is valid \emph{and} admits exact region-level execution for the whole nested family. This is also why SAT/data-light set-merging methods are not direct numerical baselines for our contour-family target: their inferential object is a one-shot set at a selected level, whereas our target is an executable merged contour whose super-level sets can be queried across levels without repeated inversion.

\section[Setup: strong validity, contours, and pointwise p-merging]{Setup: strong validity, contours, and pointwise $p$-merging}\label{sec:setup}

Let $Z$ denote observed data with distribution $P_\theta$ indexed by $\theta\in\Theta$. A \emph{plausibility function} is any mapping $\pi_Z:\Theta\to[0,1]$. Following the inferential-model convention, a (normalized) \emph{plausibility contour} additionally satisfies
\begin{equation}\label{eq:normalize}
  \sup_{\theta\in\Theta}\pi_Z(\theta)=1\qquad\text{for every realized dataset } Z=z.
\end{equation}
(In possibilistic terminology, a normalized contour is a possibility distribution and induces a possibility measure $\Pi_Z(A)=\sup_{\theta\in A}\pi_Z(\theta)$; see \cite{Zadeh1978,DuboisPrade1988}. )
In the remainder, we use \emph{$p$-value function} as the primary term and use \emph{plausibility contour} synonymously.

We allow $R$ such functions/contours $\{\pi_Z^{(r)}\}_{r=1}^R$ for the same $\theta$, obtained, for example, by repeated sample splitting, repeated randomization, or multiple algorithmic restarts. To lighten notation, we suppress the explicit dependence on $Z$ and write $\pi^{(r)}(\theta)$.

\begin{definition}[Strong validity]\label{def:sv}
A plausibility function $\pi_Z(\theta)\in[0,1]$ is \emph{strongly valid} if for all $\theta_0\in\Theta$ and all $u\in[0,1]$,
\begin{equation}\label{eq:sv}
  P_{\theta_0}\{\pi_Z(\theta_0)\le u\}\le u.
\end{equation}
\end{definition}

Strong validity ensures that the super-level set $C_\alpha(Z)=\{\theta: \pi_Z(\theta)>\alpha\}$ is a $(1-\alpha)$ confidence region. We use two standard closure facts throughout: pointwise $p$-merging preserves strong validity, and optional post-hoc normalization can preserve strong validity when a normalized contour is desired. Formal statements and proofs are collected in Appendix~\ref{app:supporting-validity}, so the main text can focus on the region-level computation problem.

\section{Grid-free inversion and finite-layer voting}\label{sec:votingframework}

Pointwise contour aggregation is statistically valid but computationally expensive: one first evaluates $\pi_{\mathrm{agg}}(\theta)$ over $\Theta$, then inverts by thresholding. In continuous or high-dimensional parameter spaces, this inversion step is often the main bottleneck. By contrast, many split-and-refit pipelines already provide split-wise regions $C_r(t)=\{\theta:\pi^{(r)}(\theta)>t\}$ at selected levels, and set operations on these regions are typically cheap. For contour-based procedures, this is a methodological issue, not a cosmetic implementation detail: a contour is useful only if its super-level sets can be queried over many $\alpha$ values without repeated dense gridding. The key question is therefore whether membership in the aggregated region can be decided from finitely many thresholded split-wise regions, without evaluating $\pi_{\mathrm{agg}}$ on a parameter grid. The next definitions formalize this implementability property. For a merged contour $\pi_{\mathrm{agg}}$, define the induced confidence region at level $\alpha\in(0,1)$ by
\[C_{\mathrm{agg}}(\alpha)=\{\theta: \pi_{\mathrm{agg}}(\theta)>\alpha\}.\]

\begin{definition}[Vote operator]\label{def:vote}
Given sets $A_1,\dots,A_R\subseteq\Theta$ and an integer $m\in\{1,\dots,R\}$, define
\[\Vote_m(A_1,\dots,A_R)=\Bigl\{\theta:\sum_{r=1}^R \ind{\theta\in A_r}\ge m\Bigr\}.
\]
\end{definition}

\paragraph*{A one-dimensional example.}
Before giving the formal definition, consider an explicit one-step-calibrator calculation with
\(R=2\), \(\Theta=\mathbb{R}\), and target level \(\alpha=0.05\). Suppose that the two
split-wise regions at the adjusted level \(t=0.025\) are
\[
  C_1(0.025)=[-1.3,0.6],\qquad C_2(0.025)=[-0.4,1.4].
\]
Take the one-step calibrator
\[
  f(u)=40\cdot\mathbf{1}_{\{u\le 0.025\}},
\]
which satisfies \(\int_0^1 f(u)\,du=40\times 0.025=1\). Writing
\(N_\theta(t)=\sum_{r=1}^2\ind{\pi^{(r)}(\theta)\le t}\), the calibrator-induced merger obeys
\[
  \pi_f(\theta)>0.05
  \quad\Longleftrightarrow\quad
  40N_\theta(0.025)<2/0.05=40
  \quad\Longleftrightarrow\quad
  N_\theta(0.025)=0.
\]
Since \(N_\theta(0.025)=0\) is equivalent to membership in both split-wise regions, the
merged region is obtained directly as
\[
  C_f(0.05)=\Vote_2(C_1(0.025),C_2(0.025))
  =C_1(0.025)\cap C_2(0.025)=[-0.4,0.6].
\]
No parameter grid is evaluated. If one instead wanted the at-least-one-split vote, the
one-step calibrator \(f(u)=20\cdot\mathbf{1}_{\{u\le 0.05\}}\) at the same \(\alpha=0.05\) gives
\(20N_\theta(0.05)<40\), i.e. \(N_\theta(0.05)\le 1\), which is equivalent to
\(\Vote_1(C_1(0.05),C_2(0.05))\). Figure~\ref{fig:one-dim-vote-example} displays the
two-split intersection case. The definitions below generalize this calculation from one
threshold and two intervals to finitely many thresholds, any number of splits, and more
general set operations.

\begin{figure}[H]
  \centering
  \includegraphics[width=0.72\linewidth]{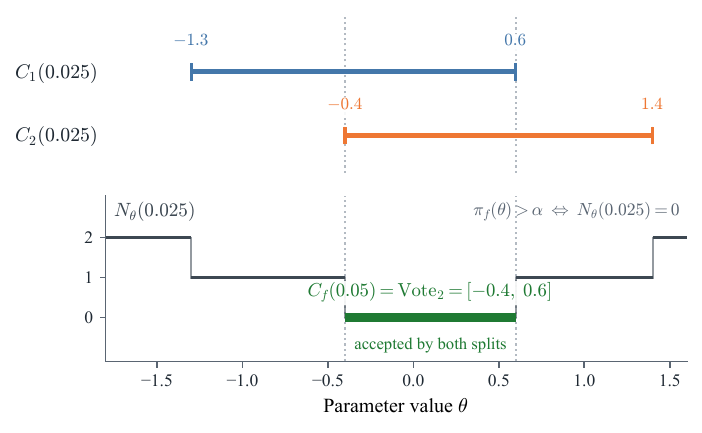}
  \caption{One-dimensional \(R=2\), \(M=1\) voting example. Top: the two split-wise intervals at the adjusted level \(t=0.025\). Bottom: the rejection count \(N_\theta(0.025)\), which changes only at the four interval endpoints; the merged region \(C_f(0.05)=\{\theta: N_\theta(0.025)=0\}=\Vote_2(C_1(0.025),C_2(0.025))\) is read off exactly, without evaluating a merged contour on a parameter grid.}
  \label{fig:one-dim-vote-example}
\end{figure}

\begin{definition}[Finite-layer vote-implementability]\label{def:vote-impl}
Let $\pi_{\mathrm{agg}}$ be a merged plausibility function built from $\{\pi^{(r)}\}_{r=1}^R$. We say that $\pi_{\mathrm{agg}}$ is \emph{$M$-layer vote-implementable} if for each $\alpha\in(0,1)$ there exist adjusted levels $\alpha_1'(\alpha),\dots,\alpha_M'(\alpha)$ such that
\[
C_{\mathrm{agg}}(\alpha)\quad\text{can be expressed using finitely many unions/intersections of vote sets}
\]
\[
\quad \Vote_{m}\bigl(C_1(\alpha_j'),\dots,C_R(\alpha_j')\bigr),\qquad j\in\{1,\dots,M\},\ m\in\{1,\dots,R\}.
\]
In particular, computing $C_{\mathrm{agg}}(\alpha)$ requires only the split-wise regions at $M$ adjusted levels and set operations, without evaluating $\pi_{\mathrm{agg}}(\theta)$ on a grid over $\Theta$.
\end{definition}

The definition allows a range of implementations, from the simplest single-threshold vote (Section~\ref{sec:orderstat}) to multi-threshold constructions (Section~\ref{sec:adaptive}). We next give a broad sufficient condition based on step calibrators.

\begin{figure}[t]
  \centering
  \includegraphics[width=0.98\linewidth]{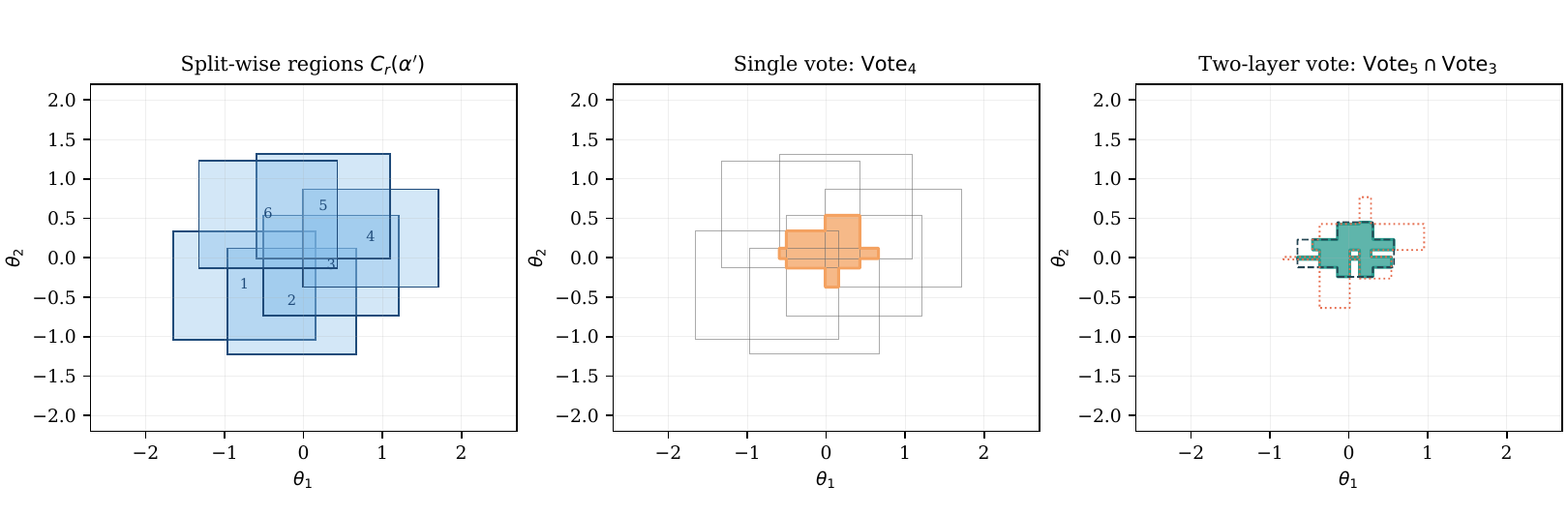}
  \caption{Schematic two-dimensional geometry of vote-based inversion. Left: split-wise regions $C_r(\alpha')$. Middle: a single-threshold vote region. Right: a two-layer vote intersection (MQ-style). The construction illustrates how contour inversion is reduced to set operations on split-wise regions.}
  \label{fig:concept-voting}
\end{figure}

\subsection{Step-calibrator induced merging and finite-layer voting}\label{subsec:stepcal}

The connection between $p$-values and e-values via \emph{calibrators} provides a flexible way to construct merging rules \cite{VovkWang2021}. We specialize it to step calibrators and show that it yields finite-layer voting.

\begin{definition}[Calibrators and step calibrators]\label{def:cal}
A measurable nonincreasing function $f:[0,1]\to[0,\infty)$ is a \emph{calibrator} if
\[\int_0^1 f(u)\,du\le 1.
\]
An \emph{$M$-step calibrator} is a calibrator of the form
\begin{equation}\label{eq:stepcal}
f(u)=\sum_{j=1}^M c_j\,\ind{u\le t_j},\qquad 0<t_1<\cdots<t_M\le 1,\ \ c_j\ge 0.
\end{equation}
The calibrator condition becomes $\sum_{j=1}^M c_j t_j\le 1$.
\end{definition}

Given $R$ $p$-value functions $\pi^{(r)}(\theta)$, define the associated e-value functions $e^{(r)}(\theta)=f(\pi^{(r)}(\theta))$ and the merged contour
\begin{equation}\label{eq:cal-merge}
\pi_f(\theta)=\min\left\{1,\ \frac{1}{\frac{1}{R}\sum_{r=1}^R e^{(r)}(\theta)}\right\}.
\end{equation}
Here and throughout, we use the convention $1/0:=\infty$, so $\pi_f(\theta)=1$ whenever the averaged e-value vanishes.
By Markov's inequality, $\pi_f(\theta_0)$ is super-uniform under $P_{\theta_0}$ whenever each input is strongly valid.

\begin{theorem}[Step calibrators yield finite-layer voting]\label{thm:step-cal}
Let $f$ be an $M$-step calibrator as in \eqref{eq:stepcal} and define $\pi_f$ by \eqref{eq:cal-merge}. Then $\pi_f$ is strongly valid. Moreover, for each $\alpha\in(0,1)$, the confidence region $C_f(\alpha)=\{\theta:\pi_f(\theta)>\alpha\}$ is $M$-layer vote-implementable.
\end{theorem}

The reason is count-based. For a step calibrator, the contour inequality depends on \(\theta\) only through the finite vector of threshold counts \(N_\theta(t_j)=\sum_{r=1}^R\ind{\pi^{(r)}(\theta)\le t_j}\). Because the step weights are nonnegative, the accepted count vectors form a downward-closed set, and each count constraint is exactly a vote event on the split-wise regions \(C_r(t_j)\). Appendix~\ref{app:stepcal-proof} gives the full proof and the explicit finite construction.

This sufficient condition is structural rather than merely algorithmic. Appendix~\ref{app:partial-converse} gives a fixed-level partial converse (Proposition~\ref{prop:step-necessity}) showing that, within the calibrator-induced class and under a mild nondegeneracy condition, pointwise-universal finite-layer vote-implementability essentially forces step structure. Thus finite-layer voting is not a generic wrapper for arbitrary smooth calibrators. Continuous calibrators can be statistically attractive, but they usually sit on the other side of the tractability--efficiency frontier because exact inversion reverts to pointwise grid evaluation.

\section{Single-threshold voting: scaled order-statistic merging}\label{sec:orderstat}

We now revisit the classical scaled order-statistic merging rule through the lens of executable contour inference. The underlying validity and sharpness statements are classical in the $p$-merging literature \cite{Ruger1978,VovkWangWang2022}; what matters here is that this familiar rule emerges as the one-threshold prototype of the broader finite-layer voting paradigm.

For each fixed $\theta\in\Theta$, let
\[\piOrd{1}(\theta)\le \piOrd{2}(\theta)\le\cdots\le \piOrd{R}(\theta)\]
denote the order statistics of $\{\pi^{(1)}(\theta),\dots,\pi^{(R)}(\theta)\}$.

\begin{definition}[$k$th order-statistic aggregation]\label{def:agg}
Fix integers $1\le k\le R$. Define
\begin{equation}\label{eq:pi-agg}
  \pi_{k,R}(\theta)=\min\left\{1,\ \frac{R}{k}\,\piOrd{k}(\theta)\right\}.
\end{equation}
\end{definition}

The rule is strongly valid under arbitrary dependence, and the scaling factor \(R/k\) is sharp within the scaled order-statistic family. These are classical facts in the $p$-merging literature; for completeness, Appendix~\ref{app:orderstat-details} records the formal statements and proofs. If normalized contours are required, the post-hoc normalization result in Proposition~\ref{prop:normalize} can be applied after aggregation.

\begin{proposition}[Exact single-threshold voting representation]\label{prop:voting}
Let $\alpha\in(0,1)$ and set $\alpha'=(k/R)\alpha$. Define split-wise regions $C_r(\alpha')=\{\theta: \pi^{(r)}(\theta)>\alpha'\}$ and the aggregated region $C_{k,R}(\alpha)=\{\theta: \pi_{k,R}(\theta)>\alpha\}$. Then
\[C_{k,R}(\alpha)=\Vote_{R-k+1}\bigl(C_1(\alpha'),\dots,C_R(\alpha')\bigr).
\]
\end{proposition}

Indeed, $\theta\in C_{k,R}(\alpha)$ iff $\piOrd{k}(\theta)>\alpha'$, which is equivalent to at least \(R-k+1\) split-wise contours exceeding \(\alpha'\) at \(\theta\).

\begin{algorithm}[t]
  \caption{Single-threshold region-level voting for $C_{k,R}(\alpha)$}\label{alg:voting}
  \begin{algorithmic}[1]
    \Require Split-wise plausibility functions $\{\pi^{(r)}\}_{r=1}^R$, target level $\alpha\in(0,1)$, parameter $k$
    \State $\alpha' \gets (k/R)\alpha$
    \For{$r=1,\dots,R$}
      \State Compute $C_r(\alpha')=\{\theta: \pi^{(r)}(\theta)>\alpha'\}$
    \EndFor
    \State \Return $C_{k,R}(\alpha)=\Vote_{R-k+1}(C_1(\alpha'),\dots,C_R(\alpha'))$
  \end{algorithmic}
\end{algorithm}

The parameter $k$ controls a stability--conservatism trade-off. The extremes $k=1$ and $k=R$ correspond to scaled intersection and union, while intermediate values yield interpretable fraction-of-splits support requirements. When dependence across splits is strong and positive, worst-case arbitrary-dependence calibration can be conservative; Section~\ref{sec:num} illustrates this phenomenon.

\section{Adaptive multi-quantile voting}\label{sec:adaptive}

Single-threshold voting requires specifying $k$. To reduce sensitivity to this choice, we consider an adaptive multi-quantile rule inspired by Hommel-type combination methods \cite{Hommel1983}. Fix $1\le M\le R$, choose $0<\lambda_1<\cdots<\lambda_M\le 1$, and define $k_m=\lceil \lambda_m R\rceil$.

\begin{definition}[Multi-quantile (generalized Hommel-type) contour aggregator]\label{def:mq}
Define
\begin{equation}\label{eq:mq}
  \pi_{\mathrm{mq}}(\theta)=\min\left\{1,\ h_M\,\min_{m=1,\dots,M}\frac{1}{\lambda_m}\,\piOrd{k_m}(\theta)\right\},
\end{equation}
where $h_M=\sum_{m=1}^M (\lambda_m-\lambda_{m-1})/\lambda_m$ with $\lambda_0=0$. Equation \eqref{eq:mq} is a sparse-grid Hommel-type $p$-merging rule valid under arbitrary dependence \citep{Hommel1983,VovkWangWang2022}; in the notation of \citet{GasparinWangRamdas2025}, it is an instance of the generalized Hommel-type \(F'_{\mathrm{GHom}}\) merger expressed in contour form. The exchangeability refinements developed there are not used for the arbitrary-dependence guarantee in this paper.
\end{definition}

The construction contains fixed-\(k\) voting as the \(M=1\) special case and recovers the classical Hommel combination on the full grid \(\lambda_m=m/R\); details are collected in Appendix~\ref{app:mq-details}. In practice, one need not use a dense quantile grid to benefit from the multi-threshold construction. Our working default choice $\lambda\in\{0.05,0.25,0.5,0.75,1.0\}$ is deliberately sparse: it covers near-intersection, lower-quartile, majority, upper-quartile, and near-union vote layers while keeping $M=5$ small. This gives MQ access to qualitatively different operating regimes without invoking the full $H_R$ construction when one wants an exact grid-free rule without dependence-specific threshold tuning. Appendix Table~\ref{tab:lambda-sens} shows that denser grids move the operating point only mildly in our benchmark.

\begin{proposition}[Multi-threshold voting representation]\label{prop:ghom-vote}
Let $\alpha\in(0,1)$ and define adjusted levels $\alpha_m'=(\alpha\lambda_m)/h_M$ for $m=1,\dots,M$. Let $C_r(\alpha_m')=\{\theta: \pi^{(r)}(\theta)>\alpha_m'\}$. Then the region $C_{\mathrm{mq}}(\alpha)=\{\theta: \pi_{\mathrm{mq}}(\theta)>\alpha\}$ satisfies
\begin{equation}\label{eq:mq-vote}
  C_{\mathrm{mq}}(\alpha)=\bigcap_{m=1}^M \Vote_{R-k_m+1}\bigl(C_1(\alpha_m'),\dots,C_R(\alpha_m')\bigr).
\end{equation}
\end{proposition}

The representation follows by rewriting the event \(\{\pi_{\mathrm{mq}}(\theta)>\alpha\}\) as the simultaneous inequalities \(\piOrd{k_m}(\theta)>\alpha_m'\), one for each quantile layer. The full algebra is given in Appendix~\ref{app:mq-details}.

\begin{algorithm}[t]
  \caption{Multi-threshold voting for $C_{\mathrm{mq}}(\alpha)$}\label{alg:multivote}
  \begin{algorithmic}[1]
    \Require Split-wise plausibility functions $\{\pi^{(r)}\}_{r=1}^R$, target level $\alpha\in(0,1)$, quantile grid $\{\lambda_m\}_{m=1}^M$
    \State Compute $h_M=\sum_{m=1}^M (\lambda_m-\lambda_{m-1})/\lambda_m$ with $\lambda_0=0$
    \For{$m=1,\dots,M$}
      \State $k_m\gets \lceil \lambda_m R\rceil$, \quad $\alpha_m'\gets (\alpha\lambda_m)/h_M$
      \For{$r=1,\dots,R$}
        \State Compute $C_r(\alpha_m')=\{\theta: \pi^{(r)}(\theta)>\alpha_m'\}$
      \EndFor
      \State $V_m \gets \Vote_{R-k_m+1}(C_1(\alpha_m'),\dots,C_R(\alpha_m'))$
    \EndFor
    \State \Return $C_{\mathrm{mq}}(\alpha)=\bigcap_{m=1}^M V_m$
  \end{algorithmic}
\end{algorithm}

Exact executability and pointwise power are separate design axes. Under arbitrary dependence, some admissible competitors can be more powerful pointwise \cite{VovkWangWang2022}, but they generally do not provide the finite collection of vote constraints in \eqref{eq:mq-vote}. Section~\ref{sec:discussion} returns to this efficiency--executability trade-off.

\section{Numerical experiments}\label{sec:num}

We evaluate the proposed voting algorithms in four complementary ways. First, a stylized split-and-refit regression problem isolates the dependence--conservatism trade-off under controlled Gaussian dependence. Second, a repeated-split conformal experiment tests whether the same logic transfers to predictive regions. Third, real-data examples on the diabetes and California housing datasets examine both repeated-split conformal prediction and split-aggregated coefficient intervals. Fourth, grid-resolution and multidimensional stress tests isolate the computational scaling of grid inversion.

The latter stress tests are computational diagnostics rather than substitutes for broad high-dimensional coverage studies; they are included to make explicit where the grid-free advantage comes from. As additional appendix-only checks, we report a non-Gaussian dependence experiment based on a \(t\)-copula with normal margins and a high-dimensional covariate conformal experiment with \(p=100\) predictors. The empirical objective is to assess whether exact vote-implementability can turn arbitrary-dependence contour merging from a valid pointwise recipe into an operational exact inference workflow, and whether adaptive multi-quantile voting provides a robust no-$k$-tuning default when dependence-specific threshold calibration is not desired. Representative results for the first benchmark appear in Tables~\ref{tab:numexp-main}--\ref{tab:numexp-runtime}. Additional main-text diagnostics and real-data summaries are reported in the remainder of this section, while full scenario grids and secondary comparator tables are deferred to Appendix~\ref{app:num-full}.

\subsection{Split-and-refit regression with repeated splitting}\label{subsec:sar}
We consider a one-dimensional target $\theta_0=0$ with $R=20$ split-wise estimators
\[
  \hat\theta^{(r)}=\theta_0+\mathrm{se}\cdot Z_r,\qquad (Z_1,\dots,Z_R)\sim N(0,\Sigma_\rho),
\]
where $\mathrm{se}=0.22$, $\Sigma_{\rho,ii}=1$, and $\Sigma_{\rho,ij}=\rho$ for $i\neq j$. We study $\rho\in\{0.1,0.3,0.5,0.7,0.9\}$ and $\alpha\in\{0.10,0.05\}$. For each split we use the two-sided Gaussian pivot, so the split-wise region at level $\alpha'$ is
\[
  C_r(\alpha')=\left[\hat\theta^{(r)}-z_{1-\alpha'/2}\,\mathrm{se},\ \hat\theta^{(r)}+z_{1-\alpha'/2}\,\mathrm{se}\right].
\]
Each $(\alpha,\rho)$ scenario uses 20000 replications with deterministic scenario-specific seeds derived from base seed 123.

Core methods are single split, fixed-$k$ order-statistic voting with $k\in\{1,5,10,15\}$ (Algorithm~\ref{alg:voting}), and adaptive multi-quantile voting (Algorithm~\ref{alg:multivote}) with the default five-point $\lambda$-grid from Section~\ref{sec:adaptive}. For $R=20$, this grid induces the internal vote layers $k_m=\lceil \lambda_m R\rceil\in\{1,5,10,15,20\}$. The stand-alone fixed-$k$ comparison reports $k\in\{1,5,10,15\}$, matching four of the five MQ layers; the endpoint $k=20$ is retained inside MQ as a near-union boundary layer but is not used as a separate fixed-$k$ comparator because it is typically very long and mainly serves as a boundary constraint within the multi-layer rule. Appendix comparators are the classical Fisher, Stouffer, and Simes pointwise merged $p$-values \cite{Fisher1932,StoufferEtAl1949,Simes1986}, inverted by a dense one-dimensional grid (801 points on an adaptive support). We report empirical coverage $P_{\theta_0}\{\theta_0\in C(\alpha)\}$, Monte Carlo confidence intervals for coverage, and mean region length $\E_{\theta_0}\{|C(\alpha)|\}$.

To complement the interval summaries, Figure~\ref{fig:numexp-contours} gives the contour-level version of the paper's main story. It shows one representative realization from the same Gaussian design at $\rho=0.5$. The thin gray curves are split-wise contours, and the colored curves are merged contours before the optional post-hoc normalization of Proposition~\ref{prop:normalize}. The distinction between vote-implementable and grid-inversion-based rules is visible: the vote-based mergers produce envelope-type contours whose $\alpha$-super-level sets can be obtained exactly by voting, whereas the continuous-calibrator comparator is smoother but available only through pointwise grid evaluation.

\begin{figure}[t]
  \centering
  \includegraphics[width=0.95\linewidth]{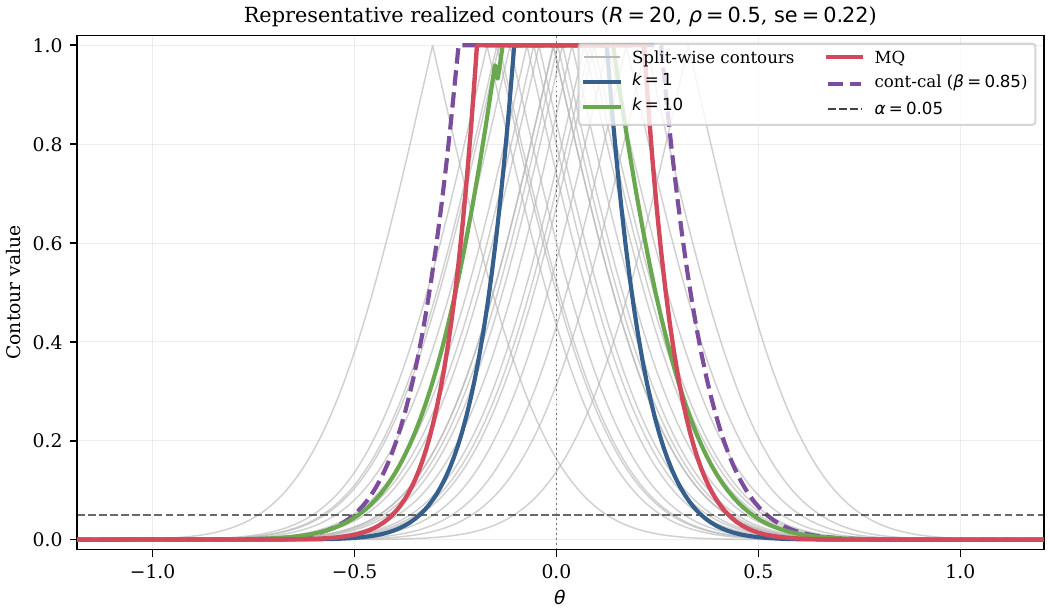}
  \caption{Representative realized split-wise and merged contours in the Gaussian split-and-refit design ($R=20$, $\rho=0.5$, $\mathrm{se}=0.22$; deterministic seed 791644). Thin gray curves are the split-wise contours. Colored curves show the raw merged contours for $k=1$, $k=10$, MQ, and the continuous-calibrator comparator. The dashed horizontal line marks $\alpha=0.05$, so its intersections with the merged curves determine the corresponding confidence regions. The merged curves are shown before the optional post-hoc normalization of Proposition~\ref{prop:normalize}.}
  \label{fig:numexp-contours}
\end{figure}

\begin{table}[t]
  \centering
  \caption{Representative results for $\alpha=0.05$ and selected dependence levels.}
  \label{tab:numexp-main}
  \begin{tabular}{ccccc}
    \toprule
    $\rho$ & Method & Coverage & Mean length & SD length \\
    \midrule
    0.1 & single & 0.9532 & 0.8624 & 0.0000 \\
    0.1 & $k=1$ & 0.9510 & 0.5522 & 0.1514 \\
    0.1 & $k=10$ & 1.0000 & 0.9605 & 0.0245 \\
    0.1 & MQ & 0.9817 & 0.6828 & 0.1477 \\
    \midrule
    0.5 & single & 0.9505 & 0.8624 & 0.0000 \\
    0.5 & $k=1$ & 0.9625 & 0.7500 & 0.1134 \\
    0.5 & $k=10$ & 0.9976 & 0.9669 & 0.0185 \\
    0.5 & MQ & 0.9856 & 0.8742 & 0.1055 \\
    \midrule
    0.9 & single & 0.9494 & 0.8624 & 0.0000 \\
    0.9 & $k=1$ & 0.9891 & 1.0707 & 0.0509 \\
    0.9 & $k=10$ & 0.9814 & 0.9776 & 0.0083 \\
    0.9 & MQ & 0.9933 & 1.1243 & 0.0225 \\
    \bottomrule
  \end{tabular}
\end{table}

\begin{table}[t]
  \centering
  \caption{Oracle-$k$ benchmark among fixed-$k$ rules and relative MQ efficiency across all $(\alpha,\rho)$ scenarios.}
  \label{tab:numexp-oracle}
  \begin{tabular}{cccccccc}
    \toprule
    $\alpha$ & $\rho$ & Oracle $k^\star$ & Oracle cov. & Oracle length & MQ cov. & MQ length & MQ/oracle \\
    \midrule
    0.05 & 0.1 & 1 & 0.9510 & 0.5522 & 0.9817 & 0.6828 & 1.2364 \\
    0.05 & 0.3 & 1 & 0.9577 & 0.6438 & 0.9848 & 0.7722 & 1.1994 \\
    0.05 & 0.5 & 1 & 0.9625 & 0.7500 & 0.9856 & 0.8742 & 1.1656 \\
    0.05 & 0.7 & 1 & 0.9740 & 0.8806 & 0.9888 & 0.9936 & 1.1283 \\
    0.05 & 0.9 & 10 & 0.9814 & 0.9776 & 0.9933 & 1.1243 & 1.1500 \\
    \midrule
    0.10 & 0.1 & 1 & 0.9060 & 0.4555 & 0.9654 & 0.5926 & 1.3012 \\
    0.10 & 0.3 & 1 & 0.9200 & 0.5496 & 0.9724 & 0.6842 & 1.2449 \\
    0.10 & 0.5 & 1 & 0.9333 & 0.6548 & 0.9743 & 0.7841 & 1.1974 \\
    0.10 & 0.7 & 1 & 0.9552 & 0.7856 & 0.9802 & 0.9009 & 1.1468 \\
    0.10 & 0.9 & 11 & 0.9601 & 0.8529 & 0.9868 & 1.0206 & 1.1966 \\
    \bottomrule
  \end{tabular}
\end{table}

\begin{table}[t]
  \centering
  \caption{Runtime benchmark (milliseconds per replication) under the same one-dimensional setup, $\alpha=0.05$, 5000 replications per $\rho$.}
  \label{tab:numexp-runtime}
  \begin{tabular}{ccccc}
    \toprule
    Method & $\rho=0.1$ & $\rho=0.5$ & $\rho=0.9$ & Mean \\
    \midrule
    single & 0.078 & 0.079 & 0.079 & 0.079 \\
    $k=1$ & 0.084 & 0.084 & 0.084 & 0.084 \\
    $k=10$ & 0.086 & 0.086 & 0.086 & 0.086 \\
    MQ & 0.410 & 0.411 & 0.414 & 0.412 \\
    Fisher & 1.206 & 1.206 & 1.201 & 1.204 \\
    Stouffer & 2.294 & 2.283 & 2.268 & 2.282 \\
    Simes & 1.018 & 1.016 & 1.008 & 1.014 \\
    \bottomrule
  \end{tabular}
\end{table}

\begin{figure}[t]
  \centering
  \includegraphics[width=0.98\linewidth]{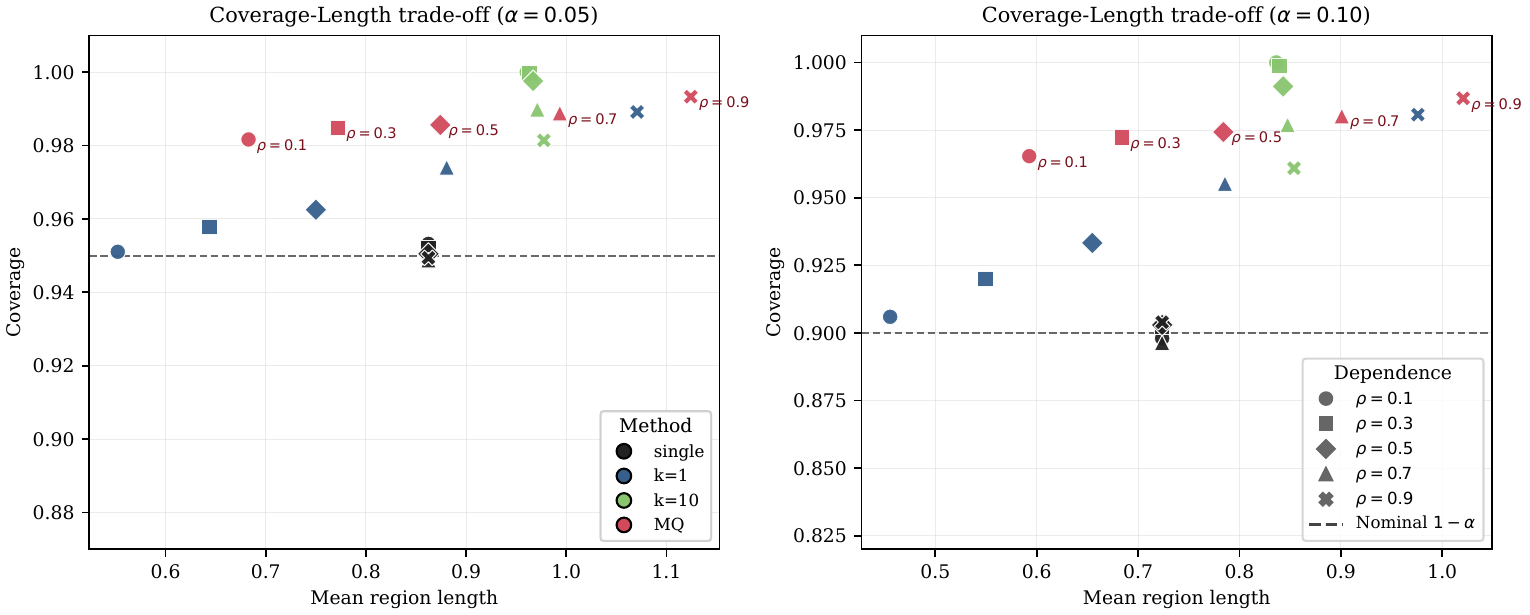}
  \caption{Coverage--length trade-off across dependence regimes in split-and-refit experiments (single, $k=1$, $k=10$, MQ). Each point is a $(\rho,\text{method})$ scenario; dashed lines mark nominal targets $1-\alpha$.}
  \label{fig:numexp-pareto}
\end{figure}

Figure~\ref{fig:numexp-contours} makes the contour-level geometry explicit: the vote-based mergers act as envelope-type operators on the split-wise curves, whereas the continuous-calibrator comparator is smoother but requires pointwise grid evaluation. Table~\ref{tab:numexp-main} confirms the expected conservatism-efficiency trade-off for fixed-$k$ voting and shows that MQ remains stable across dependence levels. Figure~\ref{fig:numexp-pareto} shows the same pattern in a concise four-method view. Relative to $k=10$, MQ reduces over-conservatism in low-to-moderate dependence while avoiding manual choice of a single fixed $k$; at very high dependence, fixed $k=10$ can be closer to nominal in this design. Table~\ref{tab:numexp-oracle} shows that the oracle choice of $k$ is dependence-sensitive: at $\alpha=0.05$ it stays at $k^\star=1$ for $\rho\le 0.7$ but shifts to $k^\star=10$ at $\rho=0.9$, and at $\alpha=0.10$ it shifts to $k^\star=11$ at $\rho=0.9$. Across all scenarios, MQ/oracle mean-length ratios lie between 1.1283 and 1.3012. We therefore interpret MQ as a robust exact default when one prefers not to commit to a dependence-specific threshold, rather than as a universally shortest rule.

We additionally assess sensitivity to the MQ quantile grid in Appendix Table~\ref{tab:lambda-sens}. The dense grid remains close to the default (coverage changes about 0.0005--0.0039 and length ratios about 1.006--1.015 across scenarios), while the coarse grid shifts the operating point more substantially. These results support the default grid as a stable operating choice.

\begin{table}[t]
  \centering
  \caption{Direct split-lottery stability diagnostics in the Gaussian split-and-refit design ($\alpha=0.05$, 20000 replications per $\rho$). Lower/upper SD are the standard deviations of the realized lower and upper endpoints across random split draws; Length IQR is the interquartile range of realized region lengths.}
  \label{tab:stability}
  \begin{tabular}{cccccc}
    \toprule
    $\rho$ & Method & Lower SD & Upper SD & Length SD & Length IQR \\
    \midrule
    0.1 & single & 0.2200 & 0.2200 & 0.0000 & 0.0000 \\
0.1 & k=1 & 0.1283 & 0.1287 & 0.1508 & 0.2018 \\
0.1 & k=10 & 0.0906 & 0.0903 & 0.0247 & 0.0286 \\
0.1 & MQ & 0.1263 & 0.1270 & 0.1471 & 0.1974 \\
0.5 & single & 0.2213 & 0.2213 & 0.0000 & 0.0000 \\
0.5 & k=1 & 0.1759 & 0.1753 & 0.1133 & 0.1502 \\
0.5 & k=10 & 0.1618 & 0.1616 & 0.0182 & 0.0212 \\
0.5 & MQ & 0.1735 & 0.1730 & 0.1055 & 0.1394 \\
0.9 & single & 0.2192 & 0.2192 & 0.0000 & 0.0000 \\
0.9 & k=1 & 0.2111 & 0.2103 & 0.0507 & 0.0681 \\
0.9 & k=10 & 0.2084 & 0.2084 & 0.0082 & 0.0094 \\
0.9 & MQ & 0.2087 & 0.2085 & 0.0222 & 0.0257 
\\
    \bottomrule
  \end{tabular}
\end{table}

Table~\ref{tab:stability} addresses the split-lottery motivation directly. In this stylized design the only randomness is the split-generated estimator draw, so endpoint dispersion across replications measures sensitivity to the split ensemble itself. MQ and $k=10$ substantially reduce endpoint variability relative to single split at low-to-moderate dependence (for example, at $\rho=0.5$, lower-endpoint SD drops from 0.2213 for single split to 0.1735 for MQ and 0.1618 for $k=10$). Single-split length is fixed by construction, so its instability appears through endpoint movement rather than length variation. These diagnostics complement the coverage/length tables by showing that aggregation stabilizes the realized region, not just its average operating characteristics.

Table~\ref{tab:numexp-runtime} quantifies the computational gain of grid-free region operations: MQ requires about 0.412 ms per replication on average, versus 1.204 ms (Fisher), 1.014 ms (Simes), and 2.282 ms (Stouffer) for grid-inversion-based comparators. The absolute times are small in this one-dimensional benchmark because the grid is deliberately modest, but the relative gap matters in the workflows motivating the paper: the same inversion is repeated across many random splits, nominal levels, target covariates, bootstrap-like repetitions, or simulation replications. Sections~\ref{subsec:grid-scaling} and~\ref{subsec:multidim-runtime} therefore report scaling checks that vary the grid resolution and dimension directly.

\begin{table}[t]
  \centering
  \caption{Comparator from a continuous nonincreasing calibrator \(f_\beta(u)=(1-\beta)u^{-\beta}\) (\(\beta=0.85\)); \(\alpha=0.05\), 20000 replications per \(\rho\).}
  \label{tab:contcal-main}
  \begin{tabular}{ccccc}
    \toprule
    $\rho$ & Method & Coverage & Mean length & SD length \\
    \midrule
    0.1 & $k=10$ & 1.0000 & 0.9605 & 0.0245 \\
    0.1 & MQ & 0.9817 & 0.6828 & 0.1477 \\
    0.1 & cont-cal($\beta$=0.85) & 0.9984 & 0.8569 & 0.1244 \\
    0.5 & $k=10$ & 0.9976 & 0.9669 & 0.0185 \\
    0.5 & MQ & 0.9856 & 0.8742 & 0.1055 \\
    0.5 & cont-cal($\beta$=0.85) & 0.9977 & 1.0244 & 0.0820 \\
    0.9 & $k=10$ & 0.9814 & 0.9776 & 0.0083 \\
    0.9 & MQ & 0.9933 & 1.1243 & 0.0225 \\
    0.9 & cont-cal($\beta$=0.85) & 0.9970 & 1.2363 & 0.0203 \\
    \bottomrule
  \end{tabular}
\end{table}

\begin{table}[t]
  \centering
  \caption{Runtime comparison for the continuous-calibrator comparator (\(\alpha=0.05\), 5000 replications per \(\rho\)).}
  \label{tab:contcal-runtime}
  \begin{tabular}{ccccc}
    \toprule
    Method & $\rho=0.1$ & $\rho=0.5$ & $\rho=0.9$ & Mean \\
    \midrule
    $k=10$ & 0.099 & 0.102 & 0.104 & 0.102 \\
    MQ & 0.423 & 0.431 & 0.436 & 0.430 \\
    cont-cal($\beta$=0.85) & 1.206 & 1.228 & 1.244 & 1.226 \\
    \bottomrule
  \end{tabular}
\end{table}

Table~\ref{tab:contcal-main} adds an arbitrary-dependence-valid comparator outside the finite-layer vote class: a continuous calibrator merger implemented by pointwise grid inversion. We fix \(\beta=0.85\) as a representative mid-high choice from the standard power-calibrator family \(f_\beta(u)=(1-\beta)u^{-\beta}\): Appendix Table~\ref{tab:contcal-beta-sens} shows that the qualitative pattern is stable over \(\beta\in\{0.65,0.75,0.85,0.95\}\), and \(\beta=0.85\) sits near the shortest intervals among the tested values in the main split-and-refit design. The comparator is generally more conservative than MQ (for example, at \(\rho=0.5\), coverage/length \(=0.9977/1.0244\) versus \(0.9856/0.8742\) for MQ). Table~\ref{tab:contcal-runtime} shows the corresponding computational cost (mean \(\approx 1.226\) ms for the continuous-calibrator comparator versus \(\approx 0.430\) ms for MQ), illustrating the efficiency--computation trade-off discussed in Section~\ref{sec:discussion}.

Assumption-sensitive comparators in Appendix Tables~\ref{tab:numexp-full-a}--\ref{tab:numexp-full-d} provide context for practical trade-offs: Fisher and Stouffer can under-cover severely under strong dependence (for example, at $\alpha=0.05,\rho=0.9$, coverage is 0.7661 and 0.6816), while Simes is generally conservative. An optional exchangeability-calibrated refinement is reported separately in Appendix~\ref{app:exch-full}: within the same vote architecture, model-based threshold calibration can move coverage closer to nominal and shorten regions, but these gains are assumption-sensitive and are not part of our arbitrary-dependence guarantee.

\subsection{Predictive-region experiment: repeated-split conformal ensembling}\label{subsec:conformal}
To test whether the same aggregation logic remains effective beyond parameter-contour inversion, we consider repeated-split conformal prediction for a scalar response at a test covariate \cite{Vovk2015,BarberCandesRamdasTibshirani2021}. For each split $r$, we fit a base regressor on a training subset and compute calibration scores
\[
  S_i^{(r)}=\left|Y_i-\hat m_r(X_i)\right|,\qquad i\in I_r^{\mathrm{cal}}.
\]
For a candidate response value $y$ at test covariate $x_0$, define
\[
  S_0^{(r)}(y)=\left|y-\hat m_r(x_0)\right|,\qquad
  \pi^{(r)}(y)=\frac{1+\sum_{i\in I_r^{\mathrm{cal}}}\ind{S_i^{(r)}\ge S_0^{(r)}(y)}}{n_{\mathrm{cal}}+1}.
\]
Under exchangeability of calibration and test samples within each split, each $\pi^{(r)}(\cdot)$ is a split-wise predictively valid $p$-value function for the target response value, so the same voting aggregators apply directly.

We use $R=20$ splits with train/calibration fraction $0.25/0.75$, two data-generating scenarios, and random test covariates: S1 (linear homoscedastic) and S2 (nonlinear heteroscedastic). We evaluate $n\in\{600,900\}$ and $\alpha\in\{0.10,0.05\}$ with 20000 replications per $(\text{scenario},n,\alpha)$. Core methods are the same as Section~\ref{subsec:sar}: single split, fixed-$k$ voting, and MQ; Appendix comparators remain Fisher/Stouffer/Simes with grid inversion.

\begin{table}[t]
  \centering
  \caption{Conformal ensembling results for $\alpha=0.05$, $n=600$ (20000 replications).}
  \label{tab:conformal-main}
  \begin{tabular}{ccccc}
    \toprule
    Scenario & Method & Coverage & Mean length & SD length \\
    \midrule
    S1 & single & 0.9531 & 2.3912 & 0.1086 \\
    S1 & $k=1$ & 0.9933 & 3.3777 & 0.2610 \\
    S1 & $k=10$ & 0.9775 & 2.7228 & 0.1171 \\
    S1 & MQ & 0.9907 & 3.1756 & 0.1588 \\
    \midrule
    S2 & single & 0.9521 & 3.3152 & 0.1815 \\
    S2 & $k=1$ & 0.9948 & 5.0231 & 0.4826 \\
    S2 & $k=10$ & 0.9762 & 3.8728 & 0.1946 \\
    S2 & MQ & 0.9915 & 4.6833 & 0.2913 \\
    \bottomrule
  \end{tabular}
\end{table}

Table~\ref{tab:conformal-main} shows that the same stability--conservatism pattern persists: fixed-$k$ and MQ remain conservative, while single split is closer to nominal. The full table in Appendix~\ref{app:conformal-full} confirms this pattern at $n=900$ and $\alpha=0.10$. In this conformal pipeline, assumption-sensitive comparators under-cover substantially (for example, at $\alpha=0.05$, $n=600$, Fisher/Stouffer cover about 0.7541/0.6499 in S2), while Simes remains conservative.

Oracle-$k$ summaries for this conformal pipeline (Appendix Table~\ref{tab:conformal-oracle}) show the same adaptivity pattern: MQ/oracle mean-length ratios range from 1.2163 to 1.3224 across the $(\text{scenario},n,\alpha)$ grid.

\begin{table}[t]
  \centering
  \caption{Conformal runtime benchmark (milliseconds per replication), $\alpha=0.05$, $n=600$, 2000 replications per scenario.}
  \label{tab:conformal-runtime}
  \begin{tabular}{cccc}
    \toprule
    Method & S1 & S2 & Mean \\
    \midrule
    single & 0.006 & 0.006 & 0.006 \\
    $k=1$ & 0.032 & 0.032 & 0.032 \\
    $k=10$ & 0.031 & 0.032 & 0.031 \\
    MQ & 0.133 & 0.134 & 0.134 \\
    Fisher & 0.453 & 0.463 & 0.458 \\
    Stouffer & 0.983 & 0.965 & 0.974 \\
    Simes & 0.230 & 0.235 & 0.232 \\
    \bottomrule
  \end{tabular}
\end{table}

Table~\ref{tab:conformal-runtime} shows that grid-free region operations remain substantially faster in this non-regression pipeline: MQ requires about 0.134 ms per replication on average, versus 0.458 ms (Fisher), 0.232 ms (Simes), and 0.974 ms (Stouffer) for grid-inversion comparators.

Appendix~\ref{app:hd-conformal} reports a separate high-dimensional covariate version of the same conformal workflow with \(n=600\), \(p=100\), sparse linear signal, AR(1) covariates, and ridge regression as the base learner. The inferential object remains a scalar prediction interval, but the base prediction problem is high-dimensional. The same pattern persists: single split is near nominal, while fixed-\(k\) voting and MQ are conservative and longer.

\subsection{Real-data illustrations: repeated-split conformal prediction and coefficient intervals}\label{subsec:realdata}
To complement synthetic designs, we evaluate the same contour-merging rules on two real regression datasets distributed through \texttt{scikit-learn} \cite{PedregosaEtAl2011}: diabetes \cite{EfronEtAl2004} (\(\mathrm{DIA}\), \(n=442\), \(p=10\)) and California housing \cite{PaceBarry1997} (\(\mathrm{CAL}\), \(n=20640\), \(p=8\)). For CAL, we use a fixed-size subsample of \(n=3000\) to keep replication cost comparable across datasets. In each replication, one observation is sampled as the test point \((X_0,Y_0)\), and the remaining observations are used for repeated split conformal ensembling with \(R=20\) splits and train/calibration fraction \(0.25/0.75\). We use a linear base regressor per split, absolute-residual conformity scores, and evaluate \(\alpha\in\{0.10,0.05\}\) over 3000 replications per dataset. Here, coverage denotes the empirical inclusion rate over repeated random choices of the held-out observation and repeated splits of the remaining fixed dataset.

Core methods are single split, fixed-\(k\) voting (\(k\in\{5,10,15\}\)), and MQ with the same \(\lambda\)-grid as above. As in the synthetic conformal experiments, the full tabular results for assumption-sensitive pointwise comparators (Fisher/Stouffer/Simes) are reported in Appendix~\ref{app:realdata-full}; Figure~\ref{fig:realdata-coverage-length} includes them only as reference points.

\begin{table}[H]
  \centering
  \caption{Real-data repeated-split conformal results at \(\alpha=0.05\) (3000 replications per dataset).}
  \label{tab:realdata-main}
  \begin{tabular}{ccccc}
    \toprule
    Dataset & Method & Coverage & Mean length & SD length \\
    \midrule
    DIA & single & 0.9500 & 2.9115 & 0.1378 \\
    DIA & $k=5$ & 0.9790 & 3.3162 & 0.1073 \\
    DIA & $k=10$ & 0.9750 & 3.2568 & 0.0613 \\
    DIA & $k=15$ & 0.9790 & 3.3132 & 0.0963 \\
    DIA & MQ & 0.9910 & 3.8367 & 0.0789 \\
    \midrule
    CAL & single & 0.9503 & 2.4993 & 0.0857 \\
    CAL & $k=5$ & 0.9817 & 3.6224 & 0.1479 \\
    CAL & $k=10$ & 0.9777 & 3.2287 & 0.0381 \\
    CAL & $k=15$ & 0.9703 & 2.9686 & 0.1167 \\
    CAL & MQ & 0.9860 & 3.8446 & 0.1032 \\
    \bottomrule
  \end{tabular}
\end{table}

\begin{figure}[H]
  \centering
  \includegraphics[width=0.95\linewidth]{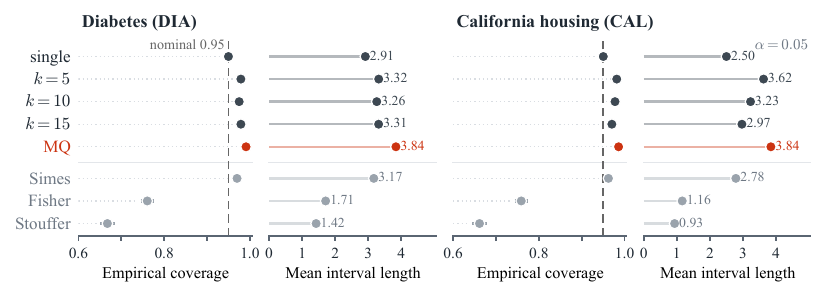}
  \caption{Coverage--length summary for the real-data repeated-split conformal experiments at \(\alpha=0.05\). Monte Carlo 95\% coverage intervals are shown and the dashed line marks nominal 0.95. Core methods correspond to Table~\ref{tab:realdata-main}; Simes, Fisher, and Stouffer are appendix comparators.}
  \label{fig:realdata-coverage-length}
\end{figure}

Table~\ref{tab:realdata-main} and Figure~\ref{fig:realdata-coverage-length} show that the same qualitative pattern appears in these real-data resampling experiments: single split is near nominal at \(\alpha=0.05\) (0.9500 on DIA; 0.9503 on CAL), while fixed-\(k\) and MQ provide conservative coverage. Among core methods, MQ attains the highest empirical coverage on both datasets and avoids pre-specifying \(k\); this adaptivity is accompanied by a visible length increase relative to the shortest fixed-\(k\) choice in each dataset. As in the synthetic experiments, we view MQ here as a robust exact no-\(k\)-tuning option rather than a length-optimal choice. Assumption-sensitive pointwise comparators under-cover substantially on both datasets (for example, at \(\alpha=0.05\), Fisher/Stouffer coverage is 0.7607/0.6677 on DIA and 0.7590/0.6617 on CAL). Full \(\alpha\in\{0.05,0.10\}\) results are reported in Appendix~\ref{app:realdata-full}.

A second real-data check, reported in Table~\ref{tab:realdata-param-main}, applies the same region aggregation idea to coefficient inference rather than prediction. Using the diabetes and California housing datasets, we fit repeated half-sample linear regressions and aggregate split-wise Gaussian coefficient intervals for one representative standardized coefficient on each dataset: \texttt{bmi} for DIA and \texttt{MedInc} for CAL. Because the true coefficient is unknown in fixed real data, this is a descriptive workflow illustration rather than a coverage study or an efficiency comparison with full-data OLS. Its purpose is to show that the same region-level aggregation machinery can be applied to split-wise coefficient intervals under a working Gaussian approximation; the full-data OLS interval is included as a familiar reference scale.

\begin{table}[H]
  \centering
  \small
  \caption{Real-data parameter-inference workflow illustration on the standardized coefficient scale (\(\alpha=0.05\), \(R=20\), deterministic split seeds 20265309 and 20270312). Descriptive fixed-data illustration only; the full-data OLS interval is included as a reference scale, not as an efficiency benchmark.}
  \label{tab:realdata-param-main}
  \begin{tabular}{ccccc}
    \toprule
    Dataset & Coefficient & Method & 95\% interval & Length \\
    \midrule
    DIA & bmi & full-data OLS & [0.241, 0.402] & 0.161 \\
DIA & bmi & single & [0.321, 0.554] & 0.233 \\
DIA & bmi & k=5 & [0.202, 0.441] & 0.239 \\
DIA & bmi & k=10 & [0.186, 0.448] & 0.262 \\
DIA & bmi & MQ & [0.239, 0.435] & 0.196 \\
CAL & MedInc & full-data OLS & [0.745, 0.818] & 0.073 \\
CAL & MedInc & single & [0.763, 0.867] & 0.104 \\
CAL & MedInc & k=5 & [0.748, 0.824] & 0.075 \\
CAL & MedInc & k=10 & [0.736, 0.831] & 0.095 \\
CAL & MedInc & MQ & [0.746, 0.806] & 0.060 
\\
    \bottomrule
  \end{tabular}
\end{table}

A matched runtime benchmark on the same real-data conformal pipeline (Appendix Table~\ref{tab:realdata-runtime}) shows the computational advantage of grid-free contour operations: MQ requires about 0.130 ms per replication on average, versus 0.362 ms (Fisher), 0.144 ms (Simes), and 1.535 ms (Stouffer). Fixed-\(k\) voting is faster still (about 0.03 ms), so MQ provides a robust, no-fixed-\(k\)-tuning option while retaining sub-millisecond runtime.

\subsection{Grid-resolution scaling for exact voting methods}\label{subsec:grid-scaling}
To isolate the role of inversion-grid resolution, we revisit the one-dimensional Gaussian split-and-refit design at \(\rho=0.5\) and \(\alpha=0.05\), varying the grid size used by the pointwise comparators over \(G\in\{201,401,801,1601\}\). Because vote-based methods operate directly on split-wise regions, their exact region computations should be essentially insensitive to \(G\), whereas pointwise grid inversion should grow with the number of evaluation points.

\begin{table}[t]
  \centering
  \caption{Runtime scaling with inversion-grid resolution in the one-dimensional Gaussian split-and-refit design (\(\rho=0.5\), \(\alpha=0.05\), 3000 replications per \(G\)). Values are milliseconds per replication.}
  \label{tab:grid-scaling-runtime}
  \begin{tabular}{cccccc}
    \toprule
    Method & \(G=201\) & \(G=401\) & \(G=801\) & \(G=1601\) & Mean \\
    \midrule
    $k=1$ & 0.089 & 0.090 & 0.095 & 0.102 & 0.094 \\
    $k=10$ & 0.085 & 0.085 & 0.087 & 0.088 & 0.086 \\
    MQ & 0.402 & 0.403 & 0.403 & 0.407 & 0.404 \\
    Fisher-grid & 0.481 & 0.711 & 1.180 & 2.127 & 1.125 \\
    Stouffer-grid & 0.816 & 1.286 & 2.239 & 4.414 & 2.189 \\
    Simes-grid & 0.348 & 0.557 & 0.995 & 2.206 & 1.026 \\
    \bottomrule
  \end{tabular}
\end{table}

Table~\ref{tab:grid-scaling-runtime} shows the expected pattern. The exact vote-based methods are nearly flat as \(G\) increases: \(k=1\) changes from about 0.089 ms to 0.102 ms, \(k=10\) from 0.085 ms to 0.088 ms, and MQ from 0.402 ms to 0.407 ms. By contrast, the pointwise grid-inversion comparators grow materially with \(G\): Fisher rises from 0.481 ms to 2.127 ms, Stouffer from 0.816 ms to 4.414 ms, and Simes from 0.348 ms to 2.206 ms. Thus the important quantity is not only the millisecond count at one grid resolution, but the slope with respect to the inversion grid. This benchmark complements Tables~\ref{tab:numexp-runtime}, \ref{tab:conformal-runtime}, and Appendix~\ref{app:realdata-runtime} by showing directly that MQ's computational advantage comes from eliminating the inversion grid rather than from a favorable fixed implementation constant.

In terms of computational complexity, fixed-\(k\) voting with interval or box primitives requires region construction at one adjusted level and vote/set operations over \(R\) split-wise regions; MQ repeats the same type of operations over \(M\) pre-specified layers. In contrast, pointwise grid inversion evaluates the merged contour on \(G^d\) candidate points in a \(d\)-dimensional Cartesian grid, with leading evaluation cost scaling with \(G^dR\) up to method-specific sorting and merging constants. The timing tables are therefore intended to illustrate this scaling contrast, not merely fixed implementation constants.

\subsection{Multidimensional runtime stress test: grid inversion}\label{subsec:multidim-runtime}
To illustrate the grid-induced dimensional blow-up highlighted in Section~\ref{sec:intro}, we add a dedicated multidimensional stress test in the simplest exact box-geometry setting. For each split $r$ and dimension $d$, we define
\[
  \pi^{(r)}(\theta)
  =
  \min\!\left\{1,\ d\cdot 2\bar\Phi\!\left(\max_{1\le j\le d}\frac{|\theta_j-\hat\theta^{(r)}_j|}{\mathrm{se}}\right)\right\},
\]
so each split region $C_r(\alpha')=\{\theta:\pi^{(r)}(\theta)>\alpha'\}$ is an axis-aligned box in $\mathbb{R}^d$. We use the single-threshold $k=1$ rule as a transparent exact baseline because its aggregated region is simply the intersection of split-wise boxes; for multi-layer rules such as MQ, the cost of exact higher-dimensional set operations becomes geometry-specific rather than grid-driven. We therefore compare this single-threshold grid-free baseline with Fisher/Stouffer/Simes pointwise merging inverted on a Cartesian grid with $G=15$ points per coordinate ($G^d$ total points). We use $R=20$, $\alpha=0.05$, $\mathrm{se}=0.22$, split-correlation parameter $\rho=0.5$, and 20 replications for each $d\in\{1,2,3,4,5\}$.

\begin{table}[t]
  \centering
  \caption{Multidimensional runtime stress test for the single-threshold grid-free baseline (milliseconds per replication).}
  \label{tab:multidim-runtime}
  \begin{tabular}{ccccccc}
    \toprule
    Method & $d=1$ & $d=2$ & $d=3$ & $d=4$ & $d=5$ & Mean \\
    \midrule
    Grid points $G^d$ ($G=15$) & 15 & 225 & 3375 & 50625 & 759375 & -- \\
    $k=1$-vote & 0.088 & 0.090 & 0.124 & 0.184 & 0.199 & 0.137 \\
    Fisher-grid & 0.265 & 0.683 & 8.634 & 155.135 & 2493.705 & 531.684 \\
    Stouffer-grid & 0.384 & 1.099 & 14.487 & 250.463 & 3906.422 & 834.571 \\
    Simes-grid & 0.156 & 0.540 & 8.442 & 152.565 & 2489.577 & 530.256 \\
    \bottomrule
  \end{tabular}
\end{table}

\begin{figure}[t]
  \centering
  \includegraphics[width=0.92\linewidth]{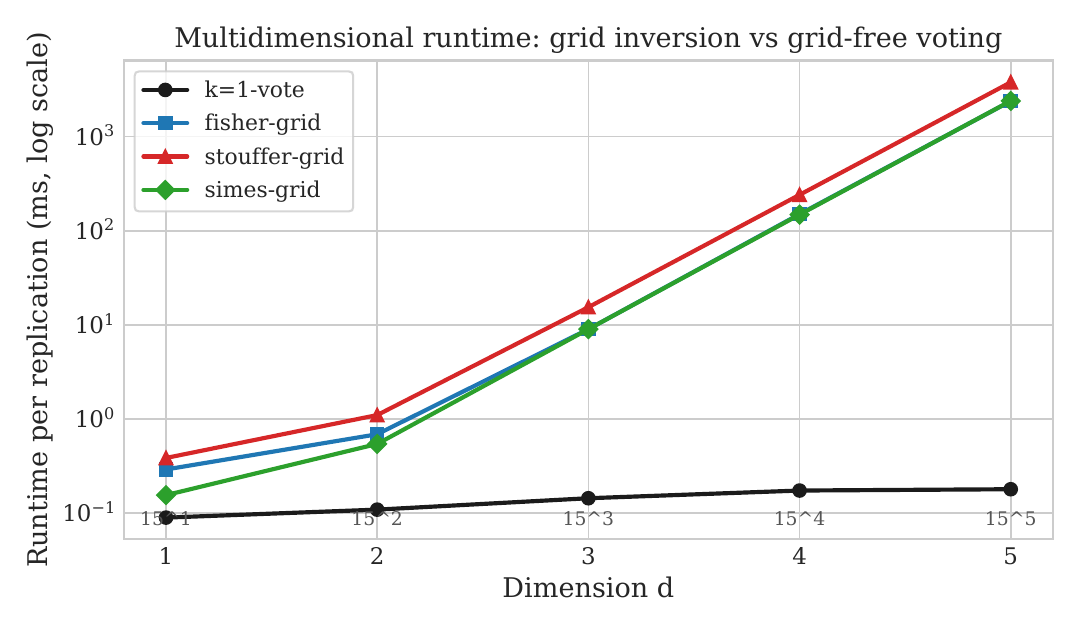}
  \caption{Runtime scaling with dimension $d$ (log scale). Grid-inversion methods increase rapidly with $G^d$, while the single-threshold grid-free vote baseline remains nearly constant.}
  \label{fig:multidim-runtime}
\end{figure}

Table~\ref{tab:multidim-runtime} and Figure~\ref{fig:multidim-runtime} illustrate the expected dimensional blow-up for pointwise inversion in this box-geometry baseline. At $d=5$, the exact single-threshold vote baseline requires about 0.199 ms per replication, versus 2493.705 ms (Fisher), 2489.577 ms (Simes), and 3906.422 ms (Stouffer), corresponding to speedups of roughly $1.25\times 10^4$--$1.96\times 10^4$. This stress test is intentionally simple, but it isolates the mechanism behind the gap: Cartesian inversion scales with \(G^d\), whereas the exact vote baseline operates on split-wise regions and, consistent with Table~\ref{tab:grid-scaling-runtime}, does not depend on the inversion grid. Full mean/median/SD runtime summaries are reported in Appendix~\ref{app:multidim-runtime-full}.

\section{Discussion}\label{sec:discussion}

The computational bottleneck addressed in this paper is repeated grid inversion of pointwise-valid merged contours. Our results show that, in a broad class of contour mergers, this bottleneck can be removed exactly by region-level vote operations. The structural theory of finite-layer vote-implementability explains when such computation is possible; the discussion below situates that perspective relative to adjacent set-merging work and summarizes its practical scope.

\subsection{Relationship to set-level voting and test inversion}
Majority vote aggregation of uncertainty sets has been studied \cite{GasparinRamdas2024}, and recent SAT/data-light work studies uncertainty-set merging through synthetic tests and test inversion \cite{QinHeGangXia2024}. Our contour-level focus complements these set-level perspectives. We study when a merged contour can be inverted to a confidence region using finitely many vote operations, and therefore when the entire nested family can be queried exactly at the region level. This is why those methods are adjacent references rather than direct numerical baselines for our contour-family target: they solve a one-level set-construction problem, whereas our main object is executable contour-family inference.

\subsection{Geometric scope of grid-free gains}
The computational gains reported here are strongest when split-wise regions admit efficient set primitives (membership testing, intersection, and union), such as intervals in one dimension or axis-aligned boxes in moderate dimensions. In more complex geometries (for example, highly nonconvex or implicitly defined regions), vote-based inversion can remain exact but the cost of the underlying set operations may become the dominant bottleneck. In that regime, the practical challenge shifts from grid evaluation of merged contours to computational-geometry routines for region operations.

\subsection{Exchangeability and randomized improvements}
Repeated randomization procedures naturally produce exchangeable (or nearly exchangeable) collections of $p$-values. Under exchangeability, many classical merging rules admit strict improvements \cite{GasparinWangRamdas2025}. Appendix~\ref{app:exch-full} gives one illustrative model-calibrated benchmark: tightening the split-level threshold within the same vote architecture can move coverage closer to nominal and reduce length, but this gain is assumption-sensitive and can deteriorate under dependence misspecification. In practice, we view this as a layered strategy: robust arbitrary-dependence voting as the default, plus model-calibrated refinements when exchangeability assumptions are well justified. Our main guarantees therefore remain anchored in arbitrary dependence.

\subsection{Practical operating guidance}
In practical terms, the proposed voting rules are most useful when split-wise intervals, boxes, or other efficiently manipulable regions are already available from repeated randomization, and the analyst must query many nominal levels, target covariates, or simulation repetitions. In such workflows, fixed-\(k\) and MQ voting replace repeated pointwise contour evaluation over an inversion grid by set operations at a small number of adjusted levels. The price is conservatism relative to dependence-specific or oracle-tuned choices, so the methods are intended as robust exact defaults rather than universally shortest procedures. The default MQ grid $\lambda\in\{0.05,0.25,0.5,0.75,1.0\}$ keeps $M$ small while spanning near-intersection, lower-quartile, majority, upper-quartile, and near-union vote layers; Appendix Table~\ref{tab:lambda-sens} shows that denser grids move the operating point little.

\begin{table}[H]
  \centering
  \footnotesize
  \setlength{\tabcolsep}{3pt}
  \renewcommand{\arraystretch}{1.08}
  \caption{Practical trade-offs among representative contour aggregation strategies.}
  \label{tab:practical-guidance}
  \begin{tabular}{L{0.14\linewidth}L{0.20\linewidth}L{0.20\linewidth}L{0.35\linewidth}}
    \toprule
    Strategy & Assumption burden & Region computation & Most appropriate use \\
    \midrule
    Fixed-\(k\) voting & Arbitrary dependence; user chooses \(k\) & Exact vote operations at one adjusted level & Short, transparent regions when a stable support threshold is justified. \\
    MQ voting & Arbitrary dependence; sparse quantile grid fixed in advance & Exact finite-layer vote operations & Robust default when avoiding dependence-specific tuning matters more than shortest length. \\
    Continuous calibrator & Arbitrary dependence; calibrator shape chosen by user & Pointwise contour evaluation followed by numerical inversion & Low-dimensional settings where smooth contours justify grid inversion. \\
    Assumption-sensitive merging & Exchangeability or model-based dependence calibration & Usually pointwise inversion or calibrated voting & Efficiency-oriented analyses with credible dependence modeling. \\
    \bottomrule
  \end{tabular}
\end{table}

Table~\ref{tab:practical-guidance} summarizes how we intend exact vote-implementability to be used. It is not proposed as a universal replacement for statistical efficiency or simplicity. Rather, it is a design criterion that becomes decisive when exact region computation, repeated queries of a nested family, and arbitrary-dependence validity are all part of the workflow. In small one-dimensional problems, a smooth continuous-calibrator merger may be preferable. In larger or repeatedly queried problems, exact vote operations remove the inversion grid and make the computational cost depend on the number of split-wise regions and vote layers rather than on the number of grid points.

\subsection{Limitations and outlook}
Two natural directions remain. The first is a deeper necessity theory beyond the fixed-level pointwise-universal partial converse in Proposition~\ref{prop:step-necessity}, especially when adjusted levels are allowed to vary with $\alpha$. The technical difficulty is that $\alpha$-dependent level maps can encode richer decision boundaries in count-space, so region-level behavior need not identify a unique calibrator shape. The second is broader empirical coverage across additional domains, including classification and higher-complexity predictive pipelines. Appendix~\ref{app:hd-conformal} adds a \(p=100\) predictive-workflow check, while broader high-dimensional parameter-region validation remains outside the scope of the present empirical study. For classification with a small discrete label space, direct label-wise evaluation is already cheap, so the main computational motivation for grid-free inversion is weaker than in continuous or high-dimensional parameter spaces.

At a higher level, the paper argues that exact executability should be considered as an additional design criterion alongside validity and region size. Pointwise merging answers whether a contour value is legitimate; finite-layer vote-implementability answers whether the induced nested family is exact, queryable, and computationally usable. In repeated-randomization settings, that distinction is operationally consequential. The step-calibrator and multi-quantile constructions developed here provide one concrete class in which these requirements can be met simultaneously.

\clearpage
\renewcommand{\appendixpagename}{Appendix}
\appendixpage
\appendix

\section{Supporting theoretical details}\label{app:theory-details}

\subsection{Validity and normalization facts}\label{app:supporting-validity}

\begin{proposition}[Post-hoc normalization preserves strong validity]\label{prop:normalize}
Let $\tilde\pi_Z(\theta)$ be strongly valid and define $M_Z=\sup_{\vartheta\in\Theta}\tilde\pi_Z(\vartheta)$. Assume $M_Z$ is measurable and finite almost surely (for example, if $\Theta$ is a separable metric space and $\theta\mapsto\tilde\pi_Z(\theta)$ is upper semicontinuous almost surely). Set $\pi_Z^\star(\theta)=\tilde\pi_Z(\theta)/M_Z$ (with $0/0=0$). Then $\pi_Z^\star$ is strongly valid and satisfies \eqref{eq:normalize} with probability one under each $P_{\theta_0}$.
\end{proposition}

\begin{proof}
Fix $\theta_0$ and $u\in[0,1]$. Strong validity at $u=0$ gives $P_{\theta_0}\{\tilde\pi_Z(\theta_0)=0\}=0$, hence $\tilde\pi_Z(\theta_0)>0$ almost surely under $P_{\theta_0}$. Since $M_Z\ge \tilde\pi_Z(\theta_0)$, we have $M_Z>0$ almost surely. On $\{M_Z>0\}$,
\[\{\pi_Z^\star(\theta_0)\le u\}=\{\tilde\pi_Z(\theta_0)\le uM_Z\}\subseteq\{\tilde\pi_Z(\theta_0)\le u\},\]
because \(\tilde\pi_Z\) takes values in \([0,1]\) and hence \(M_Z\le 1\). Therefore $P_{\theta_0}\{\pi_Z^\star(\theta_0)\le u\}\le P_{\theta_0}\{\tilde\pi_Z(\theta_0)\le u\}\le u$.
\end{proof}

\begin{proposition}[Pointwise $p$-merging yields a valid aggregated plausibility function]\label{prop:pmerge}
Let $F:[0,1]^R\to[0,1]$ be Borel measurable and satisfy: for any random vector $(U_1,\dots,U_R)$ with super-uniform marginals (arbitrary dependence allowed), $F(U_1,\dots,U_R)$ is also super-uniform. If $\{\pi^{(r)}\}_{r=1}^R$ are strongly valid, then
\[\pi_F(\theta)=F\bigl(\pi^{(1)}(\theta),\dots,\pi^{(R)}(\theta)\bigr)\]
is strongly valid.
\end{proposition}

\begin{proof}
Fix $\theta_0$. By strong validity, $(\pi^{(1)}(\theta_0),\dots,\pi^{(R)}(\theta_0))$ has super-uniform marginals. By assumption on $F$, $\pi_F(\theta_0)$ is super-uniform.
\end{proof}

\subsection{Step-calibrator vote-implementability}\label{app:stepcal-proof}

\begin{proof}[Proof of Theorem~\ref{thm:step-cal}]
Fix $\theta_0$ and write $U_r=\pi^{(r)}(\theta_0)$. By strong validity, each $U_r$ is super-uniform. For an $M$-step calibrator $f(u)=\sum_{j=1}^M c_j\,\ind{u\le t_j}$ with $\sum_{j=1}^M c_j t_j\le 1$,
\[
  \E_{\theta_0}\bigl[f(U_r)\bigr]
  =
  \sum_{j=1}^M c_j\,P_{\theta_0}\{U_r\le t_j\}
  \le
  \sum_{j=1}^M c_j t_j
  \le 1.
\]
Therefore $e^{(r)}(\theta_0)=f(U_r)$ is an e-value for each $r$, and so $\E_{\theta_0}\bigl[\frac{1}{R}\sum_{r=1}^R e^{(r)}(\theta_0)\bigr]\le 1$. For \(u=1\), strong validity is immediate because \(\pi_f(\theta_0)\le 1\). For \(u\in(0,1)\), Markov's inequality yields
\[
  P_{\theta_0}\{\pi_f(\theta_0)\le u\}
  =
  P_{\theta_0}\left\{\frac{1}{R}\sum_{r=1}^R e^{(r)}(\theta_0)\ge \frac{1}{u}\right\}
  \le u,
\]
which proves strong validity of $\pi_f$.

Now fix $\alpha\in(0,1)$ and $\theta\in\Theta$. For $j=1,\dots,M$ define $N_\theta(t_j)=\sum_{r=1}^R \ind{\pi^{(r)}(\theta)\le t_j}$. Then $\sum_{r=1}^R f(\pi^{(r)}(\theta))=\sum_{j=1}^M c_j N_\theta(t_j)$, hence
\[
  \theta\in C_f(\alpha)
  \iff
  \sum_{j=1}^M c_j N_\theta(t_j) < \frac{R}{\alpha}.
\]
Let
\[
  \mathcal{N}
  =
  \{n\in\{0,\dots,R\}^M: n_1\le\cdots\le n_M\}
\]
and
\[
  \mathcal{A}_\alpha
  =
  \left\{
    n\in\mathcal{N}:\sum_{j=1}^M c_j n_j < \frac{R}{\alpha}
  \right\}.
\]
Since $c_j\ge 0$, the set $\mathcal{A}_\alpha$ is downward closed under the coordinate-wise order. Let $\mathcal{M}_\alpha$ be the set of maximal elements of $\mathcal{A}_\alpha$. Then
\[
  \theta\in C_f(\alpha)
  \iff
  \exists\,m\in\mathcal{M}_\alpha\ \text{s.t.}\ N_\theta(t_j)\le m_j\ \text{for all } j,
\]
and therefore
\[
  C_f(\alpha)
  =
  \bigcup_{m\in\mathcal{M}_\alpha}
  \bigcap_{j=1}^M
  \{\theta: N_\theta(t_j)\le m_j\}.
\]

Finally, for each $j$ and $m_j\le R-1$,
\[
  \begin{aligned}
    \{\theta:N_\theta(t_j)\le m_j\}
    &=
    \left\{\theta:\sum_{r=1}^R \ind{\pi^{(r)}(\theta)>t_j}\ge R-m_j\right\} \\
    &=
    \Vote_{R-m_j}\bigl(C_1(t_j),\dots,C_R(t_j)\bigr),
  \end{aligned}
\]
while the case $m_j=R$ is vacuous. Hence $C_f(\alpha)$ is a finite union of intersections of vote sets at the $M$ levels $t_1,\dots,t_M$, and Definition~\ref{def:vote-impl} holds with adjusted levels $\alpha_j'(\alpha)=t_j$.
\end{proof}

The explicit construction operates in count space. The monotone count vector $\bigl(N_\theta(t_1),\dots,N_\theta(t_M)\bigr)$ takes values in a set of size $\binom{R+M}{M}$, so brute-force enumeration yields an exact implementation whose combinatorial cost depends on $(R,M)$ but not on gridding $\Theta$.

\subsection{Fixed-level partial converse}\label{app:partial-converse}

\begin{definition}[Pointwise-universal fixed-level voting]\label{def:uniform-vote}
For the partial converse below, fix a calibrator-induced merger $\pi_f$ with $R\ge 2$ and levels $0<t_1<\cdots<t_M\le 1$. For $u=(u_1,\dots,u_R)\in[0,1]^R$, define
\[
  N_j(u)=\sum_{r=1}^R \ind{u_r\le t_j},\qquad j=1,\dots,M.
\]
We say that $\pi_f$ is \emph{uniformly $M$-layer vote-implementable at $(t_1,\dots,t_M)$ in the pointwise-universal sense} if, for every $\alpha\in(0,1)$, there exists a Boolean map
\[
  \phi_\alpha:\{0,\dots,R\}^M\to\{0,1\}
\]
such that, for all $u\in[0,1]^R$,
\[
  \ind{\sum_{r=1}^R f(u_r)<R/\alpha}
  =
  \phi_\alpha\bigl(N_1(u),\dots,N_M(u)\bigr).
\]
\end{definition}

\begin{proposition}[Partial converse in the calibrator-induced class]\label{prop:step-necessity}
Fix $R\ge 2$, let $f:[0,1]\to[0,\infty)$ be nonincreasing, and define $\pi_f$ by \eqref{eq:cal-merge}. Assume there exist fixed levels $0<t_1<\cdots<t_M\le 1$ such that $\pi_f$ is uniformly $M$-layer vote-implementable at these levels in the pointwise-universal sense of Definition~\ref{def:uniform-vote}. Assume also that there exists $u_\star\in(0,1]$ with
\[
  (R-1)f(u_\star)>R.
\]
Then $f$ is constant on each right-closed interval $(t_{j-1},t_j]$ for $j=1,\dots,M$ and on $(t_M,1]$, where $t_0=0$. Equivalently, $f$ is a step function on $(0,1]$ with jump set contained in $\{t_1,\dots,t_M\}$.
\end{proposition}

This converse is structural: it rules out exact fixed-level count-based implementation of the calibrator-induced merger as a pointwise rule. It does not exclude accidental finite-level representations for a particular realized contour family whose range is restricted by the data-generating instance.

The extra condition is mild in the calibrators most relevant here. In particular, any calibrator unbounded near $0$ (such as the power family $f_\beta(u)=(1-\beta)u^{-\beta}$) automatically satisfies it for sufficiently small $u_\star$, and a step calibrator satisfies it whenever its first step height exceeds $R/(R-1)$.

\begin{proof}
Write $t_0=0$ and, for $u=(u_1,\dots,u_R)\in[0,1]^R$, define
\[
  N_j(u)=\sum_{r=1}^R \ind{u_r\le t_j},\qquad j=1,\dots,M.
\]
Let $N(u)=(N_1(u),\dots,N_M(u))$.
By pointwise-universal fixed-level vote-implementability, for each $\alpha\in(0,1)$ there exists a map
\[
  \phi_\alpha:\{0,\dots,R\}^M\to\{0,1\}
\]
such that for all $u\in[0,1]^R$,
\[
  \ind{\sum_{r=1}^R f(u_r)<R/\alpha}=\phi_\alpha(N(u)).
\]

Fix one interval in the partition induced by $(t_j)$, namely either $I=(t_{j-1},t_j]$ for some $j\in\{1,\dots,M\}$ or $I=(t_M,1]$. Take arbitrary $x,y\in I$ and define
\[
  u^{(x)}=(u_\star,\dots,u_\star,x),\qquad
  u^{(y)}=(u_\star,\dots,u_\star,y),
\]
with $R-1$ copies of $u_\star$. Since $x$ and $y$ lie in the same interval, they have identical comparisons with each threshold $t_j$, so
\[
  N(u^{(x)})=N(u^{(y)}).
\]
Therefore, for every $\alpha\in(0,1)$,
\[
  \ind{(R-1)f(u_\star)+f(x)<R/\alpha}
  =
  \ind{(R-1)f(u_\star)+f(y)<R/\alpha}.
\]
Set $B=(R-1)f(u_\star)$. By assumption, $B>R$. Since $\alpha\in(0,1)$, reparameterizing $c=R/\alpha$ ranges over all $c>R$, so the previous display is
\[
  \ind{B+f(x)<c}=\ind{B+f(y)<c}\quad\text{for all }c>R.
\]
Because $B>R$ and $f\ge 0$, both numbers $B+f(x)$ and $B+f(y)$ are strictly larger than $R$. If $B+f(x)\neq B+f(y)$, choose $c$ strictly between them; this $c$ is still $>R$, and the two indicators above differ, a contradiction. Hence $B+f(x)=B+f(y)$, so $f(x)=f(y)$.

Since $x,y\in I$ were arbitrary, $f$ is constant on $I$. Repeating for each interval gives constancy on every $(t_{j-1},t_j]$ and on $(t_M,1]$. Because the intervals $(t_{j-1},t_j]$ are right-closed, this also pins down the values at each threshold point $t_j$. Hence $f$ is a step function on $(0,1]$ with jump set contained in $\{t_1,\dots,t_M\}$.
\end{proof}

\subsection{Single-threshold voting details}\label{app:orderstat-details}

\begin{theorem}[Strong validity under arbitrary dependence]\label{thm:order-valid}
Assume each input $\pi^{(r)}$ is strongly valid. No independence across $r$ is assumed. Then $\pi_{k,R}$ is strongly valid.
\end{theorem}

\begin{proof}
Fix $\theta_0$ and $u\in[0,1)$. The event $\{\pi_{k,R}(\theta_0)\le u\}$ is equivalent to $\{\piOrd{k}(\theta_0)\le (k/R)u\}$. Let $t=(k/R)u$ and define $N(t)=\sum_{r=1}^R \ind{\pi^{(r)}(\theta_0)\le t}$. Then $\{\piOrd{k}(\theta_0)\le t\}=\{N(t)\ge k\}$. Markov's inequality yields
\[P_{\theta_0}\{N(t)\ge k\}\le \frac{\E_{\theta_0}N(t)}{k}=\frac{1}{k}\sum_{r=1}^R P_{\theta_0}\{\pi^{(r)}(\theta_0)\le t\}\le \frac{Rt}{k}=u,\]
using strong validity of each $\pi^{(r)}$.
\end{proof}

\begin{proposition}[Sharpness of the $R/k$ scaling within the scaled order-statistic family]\label{prop:sharp}
Fix $1\le k\le R$ and consider $\pi^{(c)}_{k,R}(\theta)=\min\{1,c\,\piOrd{k}(\theta)\}$ for $c>0$. If $c<R/k$, then $\pi^{(c)}_{k,R}$ fails strong validity for some super-uniform inputs. Consequently, $c=R/k$ is the smallest universal scaling in this family guaranteeing strong validity under arbitrary dependence.
\end{proposition}

\begin{proof}
Fix $c<R/k$. Choose
\[
  s\in\left(0,\min\left\{\frac{k}{R},\frac{1}{c}\right\}\right),
  \qquad
  p=\frac{Rs}{k}\in(0,1].
\]
Define an exchangeable random vector $(U_1,\dots,U_R)$ by
\[
  (U_1,\dots,U_R)=
  \begin{cases}
    (s,\dots,s,1,\dots,1) & \text{with probability } p,\\
    (1,\dots,1) & \text{with probability } 1-p,
  \end{cases}
\]
where exactly $k$ coordinates equal $s$ in the first case and all $\binom{R}{k}$ such placements are equally likely. Then each marginal is
\[
  P\{U_r\le u\}=
  \begin{cases}
    0, & u<s,\\
    pk/R=s, & s\le u<1,\\
    1, & u=1.
  \end{cases}
\]
Hence each $U_r$ is super-uniform, since $0\le u$ for $u<s$, $s\le u$ for $u\in[s,1)$, and equality holds at $u=1$.

On the event with $k$ entries at $s$, we have $U_{(k)}=s$, and otherwise $U_{(k)}=1$, hence
\[
  P\!\left\{\pi^{(c)}_{k,R}\le cs\right\}
  =
  P\{c\,U_{(k)}\le cs\}
  =
  P\{U_{(k)}\le s\}
  =
  p
  =
  \frac{Rs}{k}.
\]
Set $u_0=cs\in(0,1)$. Then
\[
  P\!\left\{\pi^{(c)}_{k,R}\le u_0\right\}
  =
  \frac{Rs}{k}
  >
  cs
  =
  u_0,
\]
so strong validity fails. Therefore no $c<R/k$ is universally valid under arbitrary dependence. This matches the classical sharpness results in \cite{Ruger1978,VovkWangWang2022}.
\end{proof}

\begin{proof}[Proof of Proposition~\ref{prop:voting}]
By definition, $\theta\in C_{k,R}(\alpha)$ iff $\piOrd{k}(\theta)>\alpha'$, which is equivalent to at least $R-k+1$ of the $R$ values $\pi^{(r)}(\theta)$ exceeding $\alpha'$.
\end{proof}

\subsection{Multi-quantile voting details}\label{app:mq-details}

\paragraph{Special cases and Hommel relation.}
If $M=1$ and $\lambda_1=k/R$, then $h_1=1$ and \eqref{eq:mq} reduces exactly to the scaled order-statistic rule \eqref{eq:pi-agg}, namely $\pi_{\mathrm{mq}}=\pi_{k,R}$. If $M=R$ and $\lambda_m=m/R$ for $m=1,\dots,R$, then $k_m=m$ and
\[
  h_M=\sum_{m=1}^R \frac{(m/R)-((m-1)/R)}{m/R}=\sum_{m=1}^R \frac{1}{m}=H_R.
\]
Hence \eqref{eq:mq} becomes
\[
  \pi_{\mathrm{mq}}(\theta)=\min\!\left\{1,\ H_R\min_{m=1,\dots,R}\frac{R}{m}\piOrd{m}(\theta)\right\},
\]
which is the classical Hommel combination (clipped at $1$). In this full-grid case, the vote thresholds in Proposition~\ref{prop:ghom-vote} simplify to $\alpha_m'=\alpha m/(R H_R)$.

\begin{proof}[Proof of Proposition~\ref{prop:ghom-vote}]
Fix $\alpha\in(0,1)$. Since $\alpha<1$, the truncation at $1$ in \eqref{eq:mq} is inactive for the event $\{\pi_{\mathrm{mq}}(\theta)>\alpha\}$, so
\[
  \pi_{\mathrm{mq}}(\theta)>\alpha
  \iff
  h_M\min_{m=1,\dots,M}\frac{\piOrd{k_m}(\theta)}{\lambda_m}>\alpha
  \iff
  \piOrd{k_m}(\theta)>\frac{\alpha\lambda_m}{h_M}\quad\text{for all }m.
\]
Thus $\theta\in C_{\mathrm{mq}}(\alpha)$ iff $\piOrd{k_m}(\theta)>\alpha_m'$ for every $m$. For each fixed $m$, this is equivalent to at least $R-k_m+1$ of the $R$ values $\pi^{(r)}(\theta)$ exceeding $\alpha_m'$, i.e., to membership in $\Vote_{R-k_m+1}(C_1(\alpha_m'),\dots,C_R(\alpha_m'))$. Intersecting over $m=1,\dots,M$ gives \eqref{eq:mq-vote}.
\end{proof}

\paragraph{Admissibility versus implementability.}
Under arbitrary dependence, many classical step-type merging rules are inadmissible and can be dominated by more powerful alternatives \cite{VovkWangWang2022}. This does not conflict with the role of finite-layer voting: admissibility concerns pointwise power, whereas vote-implementability concerns region-level computation. Multi-quantile rules yield a finite collection of vote constraints \eqref{eq:mq-vote}, while many admissible competitors do not admit such a finite-layer region representation.

\section{Full numerical tables}\label{app:num-full}

\subsection{Split-and-refit regression tables}

\clearpage
\begin{sidewaystable}[p]
  \centering
  \scriptsize
  \setlength{\tabcolsep}{12pt}
  \caption{Full Monte Carlo results across all scenarios and methods (20000 replications per scenario). Coverage CI denotes the 95\% Monte Carlo confidence interval. Part 1 of 4.}
  \label{tab:numexp-full-a}
  \begin{tabular}{ccccccccc}
    \toprule
    $\alpha$ & $\rho$ & Tier & Method & Coverage & CI low & CI high & Mean length & SD length \\
    \midrule
    0.05 & 0.1 & core & single & 0.9532 & 0.9503 & 0.9561 & 0.8624 & 0.0000 \\
0.05 & 0.1 & core & k=1 & 0.9510 & 0.9481 & 0.9540 & 0.5522 & 0.1514 \\
0.05 & 0.1 & core & k=5 & 1.0000 & 1.0000 & 1.0000 & 0.7879 & 0.0759 \\
0.05 & 0.1 & core & k=10 & 1.0000 & 1.0000 & 1.0000 & 0.9605 & 0.0245 \\
0.05 & 0.1 & core & k=15 & 1.0000 & 1.0000 & 1.0000 & 1.1618 & 0.0681 \\
0.05 & 0.1 & core & MQ & 0.9817 & 0.9798 & 0.9835 & 0.6828 & 0.1477 \\
0.05 & 0.1 & appendix & fisher & 0.9346 & 0.9312 & 0.9381 & 0.3583 & 0.0832 \\
0.05 & 0.1 & appendix & stouffer & 0.9367 & 0.9334 & 0.9401 & 0.3685 & 0.0717 \\
0.05 & 0.1 & appendix & simes & 0.9495 & 0.9464 & 0.9525 & 0.5264 & 0.1397 \\
0.05 & 0.3 & core & single & 0.9521 & 0.9491 & 0.9551 & 0.8624 & 0.0000 \\
0.05 & 0.3 & core & k=1 & 0.9577 & 0.9550 & 0.9605 & 0.6438 & 0.1338 \\
0.05 & 0.3 & core & k=5 & 0.9982 & 0.9976 & 0.9988 & 0.8247 & 0.0667 \\
0.05 & 0.3 & core & k=10 & 0.9996 & 0.9993 & 0.9999 & 0.9635 & 0.0219 \\
0.05 & 0.3 & core & k=15 & 1.0000 & 1.0000 & 1.0000 & 1.1327 & 0.0598 \\
0.05 & 0.3 & core & MQ & 0.9848 & 0.9831 & 0.9864 & 0.7722 & 0.1285 \\
0.05 & 0.3 & appendix & fisher & 0.8858 & 0.8813 & 0.8902 & 0.4070 & 0.0495 \\
0.05 & 0.3 & appendix & stouffer & 0.8798 & 0.8752 & 0.8843 & 0.4032 & 0.0409 \\
0.05 & 0.3 & appendix & simes & 0.9547 & 0.9519 & 0.9576 & 0.6136 & 0.1206 \\
0.05 & 0.5 & core & single & 0.9505 & 0.9475 & 0.9535 & 0.8624 & 0.0000 \\
0.05 & 0.5 & core & k=1 & 0.9625 & 0.9599 & 0.9651 & 0.7500 & 0.1134 \\
0.05 & 0.5 & core & k=5 & 0.9924 & 0.9912 & 0.9936 & 0.8676 & 0.0562 \\
0.05 & 0.5 & core & k=10 & 0.9976 & 0.9969 & 0.9983 & 0.9669 & 0.0185 \\
0.05 & 0.5 & core & k=15 & 0.9994 & 0.9991 & 0.9997 & 1.0986 & 0.0506 \\
0.05 & 0.5 & core & MQ & 0.9856 & 0.9839 & 0.9873 & 0.8742 & 0.1055 \\
0.05 & 0.5 & appendix & fisher & 0.8346 & 0.8294 & 0.8397 & 0.4451 & 0.0276 \\
0.05 & 0.5 & appendix & stouffer & 0.8127 & 0.8073 & 0.8182 & 0.4240 & 0.0225 \\
0.05 & 0.5 & appendix & simes & 0.9582 & 0.9554 & 0.9610 & 0.7126 & 0.0973 
\\
    \bottomrule
  \end{tabular}
\end{sidewaystable}

\begin{sidewaystable}[p]
  \centering
  \scriptsize
  \setlength{\tabcolsep}{12pt}
  \caption{Full Monte Carlo results across all scenarios and methods. Part 2 of 4.}
  \label{tab:numexp-full-b}
  \begin{tabular}{ccccccccc}
    \toprule
    $\alpha$ & $\rho$ & Tier & Method & Coverage & CI low & CI high & Mean length & SD length \\
    \midrule
    0.05 & 0.7 & core & single & 0.9487 & 0.9457 & 0.9518 & 0.8624 & 0.0000 \\
0.05 & 0.7 & core & k=1 & 0.9740 & 0.9718 & 0.9763 & 0.8806 & 0.0874 \\
0.05 & 0.7 & core & k=5 & 0.9852 & 0.9835 & 0.9869 & 0.9193 & 0.0434 \\
0.05 & 0.7 & core & k=10 & 0.9899 & 0.9885 & 0.9912 & 0.9711 & 0.0142 \\
0.05 & 0.7 & core & k=15 & 0.9947 & 0.9937 & 0.9957 & 1.0576 & 0.0391 \\
0.05 & 0.7 & core & MQ & 0.9888 & 0.9873 & 0.9903 & 0.9936 & 0.0741 \\
0.05 & 0.7 & appendix & fisher & 0.7956 & 0.7901 & 0.8012 & 0.4749 & 0.0130 \\
0.05 & 0.7 & appendix & stouffer & 0.7496 & 0.7435 & 0.7556 & 0.4298 & 0.0120 \\
0.05 & 0.7 & appendix & simes & 0.9686 & 0.9662 & 0.9710 & 0.8275 & 0.0669 \\
0.05 & 0.9 & core & single & 0.9494 & 0.9464 & 0.9524 & 0.8624 & 0.0000 \\
0.05 & 0.9 & core & k=1 & 0.9891 & 0.9877 & 0.9906 & 1.0707 & 0.0509 \\
0.05 & 0.9 & core & k=5 & 0.9830 & 0.9813 & 0.9848 & 0.9951 & 0.0255 \\
0.05 & 0.9 & core & k=10 & 0.9814 & 0.9795 & 0.9833 & 0.9776 & 0.0083 \\
0.05 & 0.9 & core & k=15 & 0.9832 & 0.9814 & 0.9850 & 0.9976 & 0.0228 \\
0.05 & 0.9 & core & MQ & 0.9933 & 0.9922 & 0.9944 & 1.1243 & 0.0225 \\
0.05 & 0.9 & appendix & fisher & 0.7661 & 0.7602 & 0.7720 & 0.4978 & 0.0039 \\
0.05 & 0.9 & appendix & stouffer & 0.6816 & 0.6752 & 0.6881 & 0.4157 & 0.0045 \\
0.05 & 0.9 & appendix & simes & 0.9765 & 0.9744 & 0.9786 & 0.9471 & 0.0179 
\\
    \bottomrule
  \end{tabular}
\end{sidewaystable}

\begin{sidewaystable}[p]
  \centering
  \scriptsize
  \setlength{\tabcolsep}{12pt}
  \caption{Full Monte Carlo results across all scenarios and methods. Part 3 of 4.}
  \label{tab:numexp-full-c}
  \begin{tabular}{ccccccccc}
    \toprule
    $\alpha$ & $\rho$ & Tier & Method & Coverage & CI low & CI high & Mean length & SD length \\
    \midrule
    0.10 & 0.1 & core & single & 0.8980 & 0.8938 & 0.9022 & 0.7237 & 0.0000 \\
0.10 & 0.1 & core & k=1 & 0.9060 & 0.9019 & 0.9100 & 0.4555 & 0.1507 \\
0.10 & 0.1 & core & k=5 & 0.9995 & 0.9992 & 0.9998 & 0.6753 & 0.0751 \\
0.10 & 0.1 & core & k=10 & 1.0000 & 1.0000 & 1.0000 & 0.8363 & 0.0248 \\
0.10 & 0.1 & core & k=15 & 1.0000 & 1.0000 & 1.0000 & 1.0298 & 0.0678 \\
0.10 & 0.1 & core & MQ & 0.9654 & 0.9628 & 0.9679 & 0.5926 & 0.1471 \\
0.10 & 0.1 & appendix & fisher & 0.8908 & 0.8865 & 0.8951 & 0.3154 & 0.0932 \\
0.10 & 0.1 & appendix & stouffer & 0.8919 & 0.8876 & 0.8962 & 0.3252 & 0.0821 \\
0.10 & 0.1 & appendix & simes & 0.9012 & 0.8971 & 0.9053 & 0.4255 & 0.1382 \\
0.10 & 0.3 & core & single & 0.9023 & 0.8982 & 0.9065 & 0.7237 & 0.0000 \\
0.10 & 0.3 & core & k=1 & 0.9200 & 0.9162 & 0.9237 & 0.5496 & 0.1338 \\
0.10 & 0.3 & core & k=5 & 0.9929 & 0.9918 & 0.9941 & 0.7114 & 0.0673 \\
0.10 & 0.3 & core & k=10 & 0.9987 & 0.9981 & 0.9992 & 0.8395 & 0.0218 \\
0.10 & 0.3 & core & k=15 & 0.9997 & 0.9995 & 0.9999 & 1.0008 & 0.0603 \\
0.10 & 0.3 & core & MQ & 0.9724 & 0.9701 & 0.9746 & 0.6842 & 0.1275 \\
0.10 & 0.3 & appendix & fisher & 0.8457 & 0.8407 & 0.8507 & 0.3713 & 0.0566 \\
0.10 & 0.3 & appendix & stouffer & 0.8375 & 0.8323 & 0.8426 & 0.3677 & 0.0498 \\
0.10 & 0.3 & appendix & simes & 0.9129 & 0.9089 & 0.9168 & 0.5145 & 0.1188 \\
0.10 & 0.5 & core & single & 0.9030 & 0.8988 & 0.9071 & 0.7237 & 0.0000 \\
0.10 & 0.5 & core & k=1 & 0.9333 & 0.9298 & 0.9367 & 0.6548 & 0.1138 \\
0.10 & 0.5 & core & k=5 & 0.9793 & 0.9773 & 0.9813 & 0.7545 & 0.0565 \\
0.10 & 0.5 & core & k=10 & 0.9911 & 0.9899 & 0.9924 & 0.8431 & 0.0185 \\
0.10 & 0.5 & core & k=15 & 0.9974 & 0.9967 & 0.9981 & 0.9666 & 0.0508 \\
0.10 & 0.5 & core & MQ & 0.9743 & 0.9721 & 0.9765 & 0.7841 & 0.1044 \\
0.10 & 0.5 & appendix & fisher & 0.8026 & 0.7971 & 0.8082 & 0.4137 & 0.0313 \\
0.10 & 0.5 & appendix & stouffer & 0.7800 & 0.7742 & 0.7857 & 0.3944 & 0.0262 \\
0.10 & 0.5 & appendix & simes & 0.9234 & 0.9197 & 0.9271 & 0.6107 & 0.0953 
\\
    \bottomrule
  \end{tabular}
\end{sidewaystable}

\begin{sidewaystable}[p]
  \centering
  \scriptsize
  \setlength{\tabcolsep}{12pt}
  \caption{Full Monte Carlo results across all scenarios and methods. Part 4 of 4.}
  \label{tab:numexp-full-d}
  \begin{tabular}{ccccccccc}
    \toprule
    $\alpha$ & $\rho$ & Tier & Method & Coverage & CI low & CI high & Mean length & SD length \\
    \midrule
    0.10 & 0.7 & core & single & 0.8963 & 0.8921 & 0.9006 & 0.7237 & 0.0000 \\
0.10 & 0.7 & core & k=1 & 0.9552 & 0.9523 & 0.9581 & 0.7856 & 0.0876 \\
0.10 & 0.7 & core & k=5 & 0.9678 & 0.9654 & 0.9703 & 0.8069 & 0.0433 \\
0.10 & 0.7 & core & k=10 & 0.9769 & 0.9749 & 0.9790 & 0.8475 & 0.0143 \\
0.10 & 0.7 & core & k=15 & 0.9865 & 0.9848 & 0.9881 & 0.9255 & 0.0391 \\
0.10 & 0.7 & core & MQ & 0.9802 & 0.9783 & 0.9821 & 0.9009 & 0.0721 \\
0.10 & 0.7 & appendix & fisher & 0.7648 & 0.7589 & 0.7707 & 0.4467 & 0.0141 \\
0.10 & 0.7 & appendix & stouffer & 0.7168 & 0.7106 & 0.7230 & 0.4046 & 0.0129 \\
0.10 & 0.7 & appendix & simes & 0.9405 & 0.9373 & 0.9438 & 0.7215 & 0.0634 \\
0.10 & 0.9 & core & single & 0.9040 & 0.8999 & 0.9080 & 0.7237 & 0.0000 \\
0.10 & 0.9 & core & k=1 & 0.9807 & 0.9788 & 0.9826 & 0.9759 & 0.0504 \\
0.10 & 0.9 & core & k=5 & 0.9674 & 0.9649 & 0.9699 & 0.8827 & 0.0251 \\
0.10 & 0.9 & core & k=10 & 0.9609 & 0.9582 & 0.9636 & 0.8538 & 0.0083 \\
0.10 & 0.9 & core & k=15 & 0.9636 & 0.9610 & 0.9661 & 0.8655 & 0.0226 \\
0.10 & 0.9 & core & MQ & 0.9868 & 0.9852 & 0.9883 & 1.0206 & 0.0190 \\
0.10 & 0.9 & appendix & fisher & 0.7421 & 0.7361 & 0.7482 & 0.4709 & 0.0040 \\
0.10 & 0.9 & appendix & stouffer & 0.6553 & 0.6488 & 0.6619 & 0.3917 & 0.0049 \\
0.10 & 0.9 & appendix & simes & 0.9531 & 0.9502 & 0.9560 & 0.8253 & 0.0142 
\\
    \bottomrule
  \end{tabular}
\end{sidewaystable}
\clearpage

\subsection{MQ quantile-grid sensitivity}\label{app:lambda-sens}
\begingroup
\scriptsize
\setlength{\tabcolsep}{4pt}
\begin{longtable}{ccccccccc}
\caption{Sensitivity of MQ to the quantile grid $\{\lambda_m\}$ in the split-and-refit design. Grids are: coarse $(0.10,0.50,1.00)$, default $(0.05,0.25,0.50,0.75,1.00)$, and dense $(0.05,0.15,0.25,0.35,0.50,0.65,0.75,0.85,1.00)$. Each scenario uses 20000 replications. Ratio is mean length relative to the default grid in the same $(\alpha,\rho)$ scenario.}\label{tab:lambda-sens}\\
\toprule
$\alpha$ & $\rho$ & Grid & Coverage & CI low & CI high & Mean len. & SD len. & Ratio \\
\midrule
\endfirsthead
\toprule
$\alpha$ & $\rho$ & Grid & Coverage & CI low & CI high & Mean len. & SD len. & Ratio \\
\midrule
\endhead
\bottomrule
\endfoot
0.05 & 0.1 & coarse & 0.9990 & 0.9986 & 0.9995 & 0.7598 & 0.1122 & 1.113 \\
0.05 & 0.1 & default & 0.9817 & 0.9798 & 0.9835 & 0.6828 & 0.1477 & 1.000 \\
0.05 & 0.1 & dense & 0.9838 & 0.9820 & 0.9855 & 0.6915 & 0.1442 & 1.013 \\
0.05 & 0.3 & coarse & 0.9957 & 0.9948 & 0.9966 & 0.8292 & 0.0986 & 1.074 \\
0.05 & 0.3 & default & 0.9848 & 0.9831 & 0.9864 & 0.7722 & 0.1285 & 1.000 \\
0.05 & 0.3 & dense & 0.9861 & 0.9845 & 0.9877 & 0.7800 & 0.1251 & 1.010 \\
0.05 & 0.5 & coarse & 0.9931 & 0.9920 & 0.9942 & 0.9093 & 0.0814 & 1.040 \\
0.05 & 0.5 & default & 0.9856 & 0.9839 & 0.9873 & 0.8742 & 0.1055 & 1.000 \\
0.05 & 0.5 & dense & 0.9874 & 0.9859 & 0.9889 & 0.8808 & 0.1022 & 1.008 \\
0.05 & 0.7 & coarse & 0.9911 & 0.9898 & 0.9924 & 1.0033 & 0.0593 & 1.010 \\
0.05 & 0.7 & default & 0.9888 & 0.9873 & 0.9903 & 0.9936 & 0.0741 & 1.000 \\
0.05 & 0.7 & dense & 0.9895 & 0.9880 & 0.9909 & 0.9996 & 0.0719 & 1.006 \\
0.05 & 0.9 & coarse & 0.9921 & 0.9909 & 0.9934 & 1.1032 & 0.0174 & 0.981 \\
0.05 & 0.9 & default & 0.9933 & 0.9922 & 0.9944 & 1.1243 & 0.0225 & 1.000 \\
0.05 & 0.9 & dense & 0.9938 & 0.9927 & 0.9949 & 1.1350 & 0.0231 & 1.009 \\
0.10 & 0.1 & coarse & 0.9954 & 0.9945 & 0.9963 & 0.6666 & 0.1115 & 1.125 \\
0.10 & 0.1 & default & 0.9654 & 0.9628 & 0.9679 & 0.5926 & 0.1471 & 1.000 \\
0.10 & 0.1 & dense & 0.9692 & 0.9668 & 0.9716 & 0.6016 & 0.1437 & 1.015 \\
0.10 & 0.3 & coarse & 0.9910 & 0.9897 & 0.9923 & 0.7368 & 0.0969 & 1.077 \\
0.10 & 0.3 & default & 0.9724 & 0.9701 & 0.9746 & 0.6842 & 0.1275 & 1.000 \\
0.10 & 0.3 & dense & 0.9751 & 0.9730 & 0.9773 & 0.6919 & 0.1240 & 1.011 \\
0.10 & 0.5 & coarse & 0.9851 & 0.9835 & 0.9868 & 0.8143 & 0.0808 & 1.038 \\
0.10 & 0.5 & default & 0.9743 & 0.9721 & 0.9765 & 0.7841 & 0.1044 & 1.000 \\
0.10 & 0.5 & dense & 0.9764 & 0.9742 & 0.9785 & 0.7911 & 0.1011 & 1.009 \\
0.10 & 0.7 & coarse & 0.9827 & 0.9808 & 0.9845 & 0.9064 & 0.0568 & 1.006 \\
0.10 & 0.7 & default & 0.9802 & 0.9783 & 0.9821 & 0.9009 & 0.0721 & 1.000 \\
0.10 & 0.7 & dense & 0.9817 & 0.9798 & 0.9835 & 0.9076 & 0.0700 & 1.007 \\
0.10 & 0.9 & coarse & 0.9845 & 0.9828 & 0.9862 & 0.9958 & 0.0143 & 0.976 \\
0.10 & 0.9 & default & 0.9868 & 0.9852 & 0.9883 & 1.0206 & 0.0190 & 1.000 \\
0.10 & 0.9 & dense & 0.9879 & 0.9863 & 0.9894 & 1.0333 & 0.0197 & 1.012 
\\
\end{longtable}
\endgroup

\subsection{Continuous-calibrator comparator under arbitrary dependence}\label{app:contcal}
\begingroup
\scriptsize
\setlength{\tabcolsep}{4pt}
\begin{longtable}{cccccccc}
\caption{Continuous-calibrator comparator \(f_\beta(u)=(1-\beta)u^{-\beta}\) with \(\beta=0.85\), compared with \(k=10\) voting and MQ in the split-and-refit design. Each scenario uses 20000 replications.}\label{tab:contcal-full}\\
\toprule
$\alpha$ & $\rho$ & Method & Coverage & CI low & CI high & Mean len. & SD len. \\
\midrule
\endfirsthead
\toprule
$\alpha$ & $\rho$ & Method & Coverage & CI low & CI high & Mean len. & SD len. \\
\midrule
\endhead
\bottomrule
\endfoot
0.05 & 0.1 & cont-cal($\beta$=0.85) & 0.9984 & 0.9979 & 0.9990 & 0.8569 & 0.1244 \\
0.05 & 0.1 & k=10 & 1.0000 & 1.0000 & 1.0000 & 0.9605 & 0.0245 \\
0.05 & 0.1 & MQ & 0.9817 & 0.9798 & 0.9835 & 0.6828 & 0.1477 \\
0.05 & 0.3 & cont-cal($\beta$=0.85) & 0.9977 & 0.9970 & 0.9984 & 0.9365 & 0.1048 \\
0.05 & 0.3 & k=10 & 0.9996 & 0.9993 & 0.9999 & 0.9635 & 0.0219 \\
0.05 & 0.3 & MQ & 0.9848 & 0.9831 & 0.9864 & 0.7722 & 0.1285 \\
0.05 & 0.5 & cont-cal($\beta$=0.85) & 0.9977 & 0.9970 & 0.9983 & 1.0244 & 0.0820 \\
0.05 & 0.5 & k=10 & 0.9976 & 0.9969 & 0.9983 & 0.9669 & 0.0185 \\
0.05 & 0.5 & MQ & 0.9856 & 0.9839 & 0.9873 & 0.8742 & 0.1055 \\
0.05 & 0.7 & cont-cal($\beta$=0.85) & 0.9966 & 0.9958 & 0.9974 & 1.1231 & 0.0543 \\
0.05 & 0.7 & k=10 & 0.9899 & 0.9885 & 0.9912 & 0.9711 & 0.0142 \\
0.05 & 0.7 & MQ & 0.9888 & 0.9873 & 0.9903 & 0.9936 & 0.0741 \\
0.05 & 0.9 & cont-cal($\beta$=0.85) & 0.9970 & 0.9962 & 0.9978 & 1.2363 & 0.0203 \\
0.05 & 0.9 & k=10 & 0.9814 & 0.9795 & 0.9833 & 0.9776 & 0.0083 \\
0.05 & 0.9 & MQ & 0.9933 & 0.9922 & 0.9944 & 1.1243 & 0.0225 \\
0.10 & 0.1 & cont-cal($\beta$=0.85) & 0.9958 & 0.9949 & 0.9967 & 0.7557 & 0.1222 \\
0.10 & 0.1 & k=10 & 1.0000 & 1.0000 & 1.0000 & 0.8363 & 0.0248 \\
0.10 & 0.1 & MQ & 0.9654 & 0.9628 & 0.9679 & 0.5926 & 0.1471 \\
0.10 & 0.3 & cont-cal($\beta$=0.85) & 0.9950 & 0.9940 & 0.9959 & 0.8358 & 0.1024 \\
0.10 & 0.3 & k=10 & 0.9987 & 0.9981 & 0.9992 & 0.8395 & 0.0218 \\
0.10 & 0.3 & MQ & 0.9724 & 0.9701 & 0.9746 & 0.6842 & 0.1275 \\
0.10 & 0.5 & cont-cal($\beta$=0.85) & 0.9939 & 0.9928 & 0.9950 & 0.9209 & 0.0799 \\
0.10 & 0.5 & k=10 & 0.9911 & 0.9899 & 0.9924 & 0.8431 & 0.0185 \\
0.10 & 0.5 & MQ & 0.9743 & 0.9721 & 0.9765 & 0.7841 & 0.1044 \\
0.10 & 0.7 & cont-cal($\beta$=0.85) & 0.9929 & 0.9917 & 0.9941 & 1.0167 & 0.0522 \\
0.10 & 0.7 & k=10 & 0.9769 & 0.9749 & 0.9790 & 0.8475 & 0.0143 \\
0.10 & 0.7 & MQ & 0.9802 & 0.9783 & 0.9821 & 0.9009 & 0.0721 \\
0.10 & 0.9 & cont-cal($\beta$=0.85) & 0.9940 & 0.9930 & 0.9951 & 1.1250 & 0.0189 \\
0.10 & 0.9 & k=10 & 0.9609 & 0.9582 & 0.9636 & 0.8538 & 0.0083 \\
0.10 & 0.9 & MQ & 0.9868 & 0.9852 & 0.9883 & 1.0206 & 0.0190 \\

\end{longtable}
\endgroup

\begin{table}[t]
  \centering
  \scriptsize
  \caption{Sensitivity of the continuous-calibrator comparator to \(\beta\) in the split-and-refit design (\(\alpha=0.05\), 10000 replications per \(\rho\)).}
  \label{tab:contcal-beta-sens}
  \begin{tabular}{ccccccc}
    \toprule
    $\rho$ & $\beta$ & Coverage & CI low & CI high & Mean length & SD length \\
    \midrule
    0.1 & 0.65 & 0.9992 & 0.9986 & 0.9998 & 0.9709 & 0.1148 \\
0.1 & 0.75 & 0.9985 & 0.9977 & 0.9993 & 0.8930 & 0.1223 \\
0.1 & 0.85 & 0.9981 & 0.9972 & 0.9990 & 0.8553 & 0.1239 \\
0.1 & 0.95 & 0.9985 & 0.9977 & 0.9993 & 0.8975 & 0.1314 \\
0.5 & 0.65 & 0.9990 & 0.9984 & 0.9996 & 1.1263 & 0.0731 \\
0.5 & 0.75 & 0.9984 & 0.9976 & 0.9992 & 1.0562 & 0.0778 \\
0.5 & 0.85 & 0.9976 & 0.9966 & 0.9986 & 1.0242 & 0.0816 \\
0.5 & 0.95 & 0.9977 & 0.9968 & 0.9986 & 1.0717 & 0.0879 \\
0.9 & 0.65 & 0.9980 & 0.9971 & 0.9989 & 1.3099 & 0.0165 \\
0.9 & 0.75 & 0.9964 & 0.9952 & 0.9976 & 1.2541 & 0.0182 \\
0.9 & 0.85 & 0.9957 & 0.9944 & 0.9970 & 1.2363 & 0.0201 \\
0.9 & 0.95 & 0.9980 & 0.9971 & 0.9989 & 1.3003 & 0.0232 
\\
    \bottomrule
  \end{tabular}
\end{table}

\subsection{Exchangeability-sensitive split-and-refit benchmark}\label{app:exch-full}
\begingroup
\scriptsize
\setlength{\tabcolsep}{4pt}
\begin{longtable}{ccccccccc}
\caption{Exchangeability-sensitive split-and-refit benchmark under the equicorrelated Gaussian model. Each scenario uses 20000 replications; calibrated thresholds $t_{10,\alpha,\rho}$ are estimated from 120000 null draws.}\label{tab:exch-full}\\
\toprule
$\alpha$ & $\rho$ & Method & Coverage & CI low & CI high & Mean len. & SD len. & $t_{10,\alpha,\rho}$ \\
\midrule
\endfirsthead
\toprule
$\alpha$ & $\rho$ & Method & Coverage & CI low & CI high & Mean len. & SD len. & $t_{10,\alpha,\rho}$ \\
\midrule
\endhead
\bottomrule
\endfoot
0.05 & 0.1 & single & 0.9532 & 0.9503 & 0.9561 & 0.8624 & 0.0000 & 0.2929 \\
0.05 & 0.1 & k=10 & 1.0000 & 1.0000 & 1.0000 & 0.9605 & 0.0245 & 0.2929 \\
0.05 & 0.1 & MQ & 0.9817 & 0.9798 & 0.9835 & 0.6828 & 0.1477 & 0.2929 \\
0.05 & 0.1 & $k=10$ (exch-cal) & 0.9474 & 0.9443 & 0.9505 & 0.4093 & 0.0544 & 0.2929 \\
0.05 & 0.3 & single & 0.9521 & 0.9491 & 0.9551 & 0.8624 & 0.0000 & 0.2186 \\
0.05 & 0.3 & k=10 & 0.9996 & 0.9993 & 0.9999 & 0.9635 & 0.0219 & 0.2186 \\
0.05 & 0.3 & MQ & 0.9848 & 0.9831 & 0.9864 & 0.7722 & 0.1285 & 0.2186 \\
0.05 & 0.3 & $k=10$ (exch-cal) & 0.9517 & 0.9488 & 0.9547 & 0.5157 & 0.0250 & 0.2186 \\
0.05 & 0.5 & single & 0.9505 & 0.9475 & 0.9535 & 0.8624 & 0.0000 & 0.1374 \\
0.05 & 0.5 & k=10 & 0.9976 & 0.9969 & 0.9983 & 0.9669 & 0.0185 & 0.1374 \\
0.05 & 0.5 & MQ & 0.9856 & 0.9839 & 0.9873 & 0.8742 & 0.1055 & 0.1374 \\
0.05 & 0.5 & $k=10$ (exch-cal) & 0.9496 & 0.9466 & 0.9527 & 0.6342 & 0.0186 & 0.1374 \\
0.05 & 0.7 & single & 0.9487 & 0.9457 & 0.9518 & 0.8624 & 0.0000 & 0.0895 \\
0.05 & 0.7 & k=10 & 0.9899 & 0.9885 & 0.9912 & 0.9711 & 0.0142 & 0.0895 \\
0.05 & 0.7 & MQ & 0.9888 & 0.9873 & 0.9903 & 0.9936 & 0.0741 & 0.0895 \\
0.05 & 0.7 & $k=10$ (exch-cal) & 0.9493 & 0.9462 & 0.9523 & 0.7320 & 0.0142 & 0.0895 \\
0.05 & 0.9 & single & 0.9494 & 0.9464 & 0.9524 & 0.8624 & 0.0000 & 0.0591 \\
0.05 & 0.9 & k=10 & 0.9814 & 0.9795 & 0.9833 & 0.9776 & 0.0083 & 0.0591 \\
0.05 & 0.9 & MQ & 0.9933 & 0.9922 & 0.9944 & 1.1243 & 0.0225 & 0.0591 \\
0.05 & 0.9 & $k=10$ (exch-cal) & 0.9500 & 0.9470 & 0.9531 & 0.8217 & 0.0083 & 0.0591 \\
0.10 & 0.1 & single & 0.8980 & 0.8938 & 0.9022 & 0.7237 & 0.0000 & 0.3328 \\
0.10 & 0.1 & k=10 & 1.0000 & 1.0000 & 1.0000 & 0.8363 & 0.0248 & 0.3328 \\
0.10 & 0.1 & MQ & 0.9654 & 0.9628 & 0.9679 & 0.5926 & 0.1471 & 0.3328 \\
0.10 & 0.1 & $k=10$ (exch-cal) & 0.9027 & 0.8986 & 0.9068 & 0.3558 & 0.0666 & 0.3328 \\
0.10 & 0.3 & single & 0.9023 & 0.8982 & 0.9065 & 0.7237 & 0.0000 & 0.2849 \\
0.10 & 0.3 & k=10 & 0.9987 & 0.9981 & 0.9992 & 0.8395 & 0.0218 & 0.2849 \\
0.10 & 0.3 & MQ & 0.9724 & 0.9701 & 0.9746 & 0.6842 & 0.1275 & 0.2849 \\
0.10 & 0.3 & $k=10$ (exch-cal) & 0.9011 & 0.8970 & 0.9052 & 0.4386 & 0.0318 & 0.2849 \\
0.10 & 0.5 & single & 0.9030 & 0.8988 & 0.9071 & 0.7237 & 0.0000 & 0.2109 \\
0.10 & 0.5 & k=10 & 0.9911 & 0.9899 & 0.9924 & 0.8431 & 0.0185 & 0.2109 \\
0.10 & 0.5 & MQ & 0.9743 & 0.9721 & 0.9765 & 0.7841 & 0.1044 & 0.2109 \\
0.10 & 0.5 & $k=10$ (exch-cal) & 0.8996 & 0.8954 & 0.9038 & 0.5308 & 0.0189 & 0.2109 \\
0.10 & 0.7 & single & 0.8963 & 0.8921 & 0.9006 & 0.7237 & 0.0000 & 0.1502 \\
0.10 & 0.7 & k=10 & 0.9769 & 0.9749 & 0.9790 & 0.8475 & 0.0143 & 0.1502 \\
0.10 & 0.7 & MQ & 0.9802 & 0.9783 & 0.9821 & 0.9009 & 0.0721 & 0.1502 \\
0.10 & 0.7 & $k=10$ (exch-cal) & 0.8970 & 0.8928 & 0.9012 & 0.6181 & 0.0143 & 0.1502 \\
0.10 & 0.9 & single & 0.9040 & 0.8999 & 0.9080 & 0.7237 & 0.0000 & 0.1121 \\
0.10 & 0.9 & k=10 & 0.9609 & 0.9582 & 0.9636 & 0.8538 & 0.0083 & 0.1121 \\
0.10 & 0.9 & MQ & 0.9868 & 0.9852 & 0.9883 & 1.0206 & 0.0190 & 0.1121 \\
0.10 & 0.9 & $k=10$ (exch-cal) & 0.9011 & 0.8969 & 0.9052 & 0.6905 & 0.0083 & 0.1121 \\

\end{longtable}
\endgroup

\subsection{Repeated-split conformal ensembling tables}\label{app:conformal-full}

\begingroup
\scriptsize
\setlength{\tabcolsep}{3pt}
\begin{longtable}{cccccccccc}
\caption{Full conformal-ensembling Monte Carlo results across scenarios, sample sizes, and methods (20000 replications per $(\text{scenario},n,\alpha)$ setting). Coverage CI denotes the 95\% Monte Carlo confidence interval.}\label{tab:conformal-full}\\
\toprule
$\alpha$ & Scen. & $n$ & Tier & Method & Coverage & CI low & CI high & Mean len. & SD len. \\
\midrule
\endfirsthead
\toprule
$\alpha$ & Scen. & $n$ & Tier & Method & Coverage & CI low & CI high & Mean len. & SD len. \\
\midrule
\endhead
\bottomrule
\endfoot
0.05 & S1 & 600 & core & single & 0.9531 & 0.9502 & 0.9560 & 2.3912 & 0.1086 \\
0.05 & S1 & 600 & core & k=1 & 0.9933 & 0.9922 & 0.9945 & 3.3777 & 0.2610 \\
0.05 & S1 & 600 & core & k=5 & 0.9858 & 0.9841 & 0.9874 & 2.9826 & 0.1526 \\
0.05 & S1 & 600 & core & k=10 & 0.9775 & 0.9754 & 0.9795 & 2.7228 & 0.1171 \\
0.05 & S1 & 600 & core & k=15 & 0.9725 & 0.9702 & 0.9748 & 2.6313 & 0.1064 \\
0.05 & S1 & 600 & core & MQ & 0.9907 & 0.9894 & 0.9920 & 3.1756 & 0.1588 \\
0.05 & S1 & 600 & appendix & fisher & 0.7530 & 0.7471 & 0.7590 & 1.3969 & 0.0506 \\
0.05 & S1 & 600 & appendix & stouffer & 0.6486 & 0.6420 & 0.6553 & 1.1199 & 0.0441 \\
0.05 & S1 & 600 & appendix & simes & 0.9683 & 0.9658 & 0.9707 & 2.5665 & 0.1011 \\
0.10 & S1 & 600 & core & single & 0.9018 & 0.8976 & 0.9059 & 1.9933 & 0.0844 \\
0.10 & S1 & 600 & core & k=1 & 0.9894 & 0.9880 & 0.9908 & 3.1462 & 0.2005 \\
0.10 & S1 & 600 & core & k=5 & 0.9723 & 0.9700 & 0.9746 & 2.6287 & 0.1142 \\
0.10 & S1 & 600 & core & k=10 & 0.9529 & 0.9500 & 0.9558 & 2.3764 & 0.0909 \\
0.10 & S1 & 600 & core & k=15 & 0.9384 & 0.9351 & 0.9417 & 2.2439 & 0.0837 \\
0.10 & S1 & 600 & core & MQ & 0.9801 & 0.9782 & 0.9820 & 2.7980 & 0.1171 \\
0.10 & S1 & 600 & appendix & fisher & 0.7285 & 0.7223 & 0.7347 & 1.3231 & 0.0489 \\
0.10 & S1 & 600 & appendix & stouffer & 0.6172 & 0.6105 & 0.6240 & 1.0506 & 0.0426 \\
0.10 & S1 & 600 & appendix & simes & 0.9307 & 0.9272 & 0.9343 & 2.1808 & 0.0832 \\
0.05 & S2 & 600 & core & single & 0.9521 & 0.9491 & 0.9551 & 3.3152 & 0.1815 \\
0.05 & S2 & 600 & core & k=1 & 0.9948 & 0.9937 & 0.9958 & 5.0231 & 0.4826 \\
0.05 & S2 & 600 & core & k=5 & 0.9861 & 0.9845 & 0.9877 & 4.3279 & 0.2634 \\
0.05 & S2 & 600 & core & k=10 & 0.9762 & 0.9741 & 0.9784 & 3.8728 & 0.1946 \\
0.05 & S2 & 600 & core & k=15 & 0.9715 & 0.9692 & 0.9738 & 3.7077 & 0.1728 \\
0.05 & S2 & 600 & core & MQ & 0.9915 & 0.9902 & 0.9927 & 4.6833 & 0.2913 \\
0.05 & S2 & 600 & appendix & fisher & 0.7541 & 0.7482 & 0.7601 & 1.8219 & 0.0725 \\
0.05 & S2 & 600 & appendix & stouffer & 0.6499 & 0.6432 & 0.6565 & 1.4386 & 0.0614 \\
0.05 & S2 & 600 & appendix & simes & 0.9676 & 0.9651 & 0.9701 & 3.5931 & 0.1608 \\
0.10 & S2 & 600 & core & single & 0.9015 & 0.8974 & 0.9057 & 2.6918 & 0.1287 \\
0.10 & S2 & 600 & core & k=1 & 0.9919 & 0.9906 & 0.9931 & 4.6063 & 0.3526 \\
0.10 & S2 & 600 & core & k=5 & 0.9711 & 0.9688 & 0.9734 & 3.7157 & 0.1853 \\
0.10 & S2 & 600 & core & k=10 & 0.9527 & 0.9498 & 0.9556 & 3.2910 & 0.1424 \\
0.10 & S2 & 600 & core & k=15 & 0.9392 & 0.9358 & 0.9425 & 3.0750 & 0.1278 \\
0.10 & S2 & 600 & core & MQ & 0.9792 & 0.9773 & 0.9812 & 3.9907 & 0.1957 \\
0.10 & S2 & 600 & appendix & fisher & 0.7298 & 0.7237 & 0.7360 & 1.7182 & 0.0692 \\
0.10 & S2 & 600 & appendix & stouffer & 0.6200 & 0.6133 & 0.6267 & 1.3445 & 0.0590 \\
0.10 & S2 & 600 & appendix & simes & 0.9327 & 0.9292 & 0.9362 & 2.9741 & 0.1261 \\
0.05 & S1 & 900 & core & single & 0.9493 & 0.9463 & 0.9523 & 2.3812 & 0.0880 \\
0.05 & S1 & 900 & core & k=1 & 0.9966 & 0.9958 & 0.9974 & 3.5365 & 0.2522 \\
0.05 & S1 & 900 & core & k=5 & 0.9858 & 0.9842 & 0.9874 & 2.9585 & 0.1238 \\
0.05 & S1 & 900 & core & k=10 & 0.9750 & 0.9729 & 0.9772 & 2.7280 & 0.0976 \\
0.05 & S1 & 900 & core & k=15 & 0.9659 & 0.9634 & 0.9684 & 2.5852 & 0.0855 \\
0.05 & S1 & 900 & core & MQ & 0.9905 & 0.9892 & 0.9918 & 3.1069 & 0.1261 \\
0.05 & S1 & 900 & appendix & fisher & 0.7512 & 0.7453 & 0.7572 & 1.3931 & 0.0417 \\
0.05 & S1 & 900 & appendix & stouffer & 0.6464 & 0.6398 & 0.6531 & 1.1151 & 0.0367 \\
0.05 & S1 & 900 & appendix & simes & 0.9619 & 0.9593 & 0.9646 & 2.5324 & 0.0830 \\
0.10 & S1 & 900 & core & single & 0.8991 & 0.8949 & 0.9032 & 1.9903 & 0.0679 \\
0.10 & S1 & 900 & core & k=1 & 0.9919 & 0.9907 & 0.9932 & 3.1827 & 0.1682 \\
0.10 & S1 & 900 & core & k=5 & 0.9699 & 0.9675 & 0.9723 & 2.6508 & 0.0956 \\
0.10 & S1 & 900 & core & k=10 & 0.9491 & 0.9461 & 0.9521 & 2.3703 & 0.0740 \\
0.10 & S1 & 900 & core & k=15 & 0.9335 & 0.9300 & 0.9370 & 2.2174 & 0.0672 \\
0.10 & S1 & 900 & core & MQ & 0.9761 & 0.9740 & 0.9783 & 2.7490 & 0.0953 \\
0.10 & S1 & 900 & appendix & fisher & 0.7269 & 0.7208 & 0.7331 & 1.3195 & 0.0403 \\
0.10 & S1 & 900 & appendix & stouffer & 0.6153 & 0.6086 & 0.6221 & 1.0461 & 0.0357 \\
0.10 & S1 & 900 & appendix & simes & 0.9244 & 0.9207 & 0.9280 & 2.1483 & 0.0686 \\
0.05 & S2 & 900 & core & single & 0.9505 & 0.9475 & 0.9535 & 3.2980 & 0.1458 \\
0.05 & S2 & 900 & core & k=1 & 0.9954 & 0.9945 & 0.9964 & 5.3305 & 0.4801 \\
0.05 & S2 & 900 & core & k=5 & 0.9846 & 0.9828 & 0.9863 & 4.2867 & 0.2120 \\
0.05 & S2 & 900 & core & k=10 & 0.9750 & 0.9729 & 0.9772 & 3.8796 & 0.1623 \\
0.05 & S2 & 900 & core & k=15 & 0.9674 & 0.9649 & 0.9698 & 3.6295 & 0.1386 \\
0.05 & S2 & 900 & core & MQ & 0.9885 & 0.9870 & 0.9899 & 4.5522 & 0.2248 \\
0.05 & S2 & 900 & appendix & fisher & 0.7528 & 0.7468 & 0.7588 & 1.8174 & 0.0602 \\
0.05 & S2 & 900 & appendix & stouffer & 0.6432 & 0.6366 & 0.6499 & 1.4332 & 0.0517 \\
0.05 & S2 & 900 & appendix & simes & 0.9637 & 0.9612 & 0.9663 & 3.5364 & 0.1322 \\
0.10 & S2 & 900 & core & single & 0.9000 & 0.8958 & 0.9042 & 2.6881 & 0.1041 \\
0.10 & S2 & 900 & core & k=1 & 0.9899 & 0.9885 & 0.9913 & 4.6725 & 0.2972 \\
0.10 & S2 & 900 & core & k=5 & 0.9710 & 0.9687 & 0.9733 & 3.7510 & 0.1561 \\
0.10 & S2 & 900 & core & k=10 & 0.9506 & 0.9476 & 0.9536 & 3.2803 & 0.1176 \\
0.10 & S2 & 900 & core & k=15 & 0.9352 & 0.9318 & 0.9386 & 3.0338 & 0.1049 \\
0.10 & S2 & 900 & core & MQ & 0.9760 & 0.9739 & 0.9781 & 3.9040 & 0.1580 \\
0.10 & S2 & 900 & appendix & fisher & 0.7255 & 0.7193 & 0.7316 & 1.7144 & 0.0580 \\
0.10 & S2 & 900 & appendix & stouffer & 0.6151 & 0.6084 & 0.6218 & 1.3392 & 0.0495 \\
0.10 & S2 & 900 & appendix & simes & 0.9271 & 0.9234 & 0.9307 & 2.9246 & 0.1060 \\

\end{longtable}
\endgroup

\subsection{High-dimensional covariate conformal check}\label{app:hd-conformal}
As an additional predictive-workflow check, we repeat the conformal experiment with \(n=600\) training observations and \(p=100\) covariates. Covariates follow a Gaussian AR(1) design with correlation \(\rho_X=0.5\), the regression coefficient has 10 nonzero entries, and the base learner is ridge regression with fixed penalty 10. The two scenarios are homoscedastic sparse linear regression (HD1) and heteroscedastic sparse linear regression (HD2). The conformal output remains a scalar prediction interval for \(Y_0\), so this check stresses the covariate side of the prediction workflow rather than high-dimensional region geometry.

\begingroup
\tiny
\setlength{\tabcolsep}{4pt}
\begin{longtable}{cccccccc}
\caption{High-dimensional covariate repeated-split conformal experiment (\(n=600\), \(p=100\), \(R=20\), train/calibration fraction \(0.25/0.75\), 2000 replications per scenario). Coverage CI denotes the 95\% Monte Carlo confidence interval. Time is aggregation time in milliseconds per replication after the split-wise conformal summaries are available.}\label{tab:hd-conformal}\\
\toprule
Scen. & \(\alpha\) & Method & Coverage & Coverage CI & Mean len. & SD len. & Time ms \\
\midrule
\endfirsthead
\toprule
Scen. & \(\alpha\) & Method & Coverage & Coverage CI & Mean len. & SD len. & Time ms \\
\midrule
\endhead
\bottomrule
\endfoot
HD1 & 0.05 & single & 0.9480 & [0.9383,0.9577] & 5.8653 & 0.3933 & 0.005 \\
HD1 & 0.05 & k=5 & 0.9890 & [0.9844,0.9936] & 6.1922 & 0.3919 & 0.060 \\
HD1 & 0.05 & k=10 & 0.9930 & [0.9893,0.9967] & 6.5721 & 0.2718 & 0.053 \\
HD1 & 0.05 & k=15 & 0.9980 & [0.9960,1.0000] & 7.3666 & 0.3912 & 0.054 \\
HD1 & 0.05 & mq & 0.9990 & [0.9976,1.0000] & 7.7501 & 0.4102 & 0.240 \\
HD1 & 0.05 & fisher & 0.8120 & [0.7949,0.8291] & 3.1280 & 0.1589 & 0.583 \\
HD1 & 0.05 & stouffer & 0.7795 & [0.7613,0.7977] & 2.9028 & 0.1364 & 2.303 \\
HD1 & 0.05 & simes & 0.9680 & [0.9603,0.9757] & 5.3464 & 0.6183 & 0.221 \\
HD1 & 0.10 & single & 0.8975 & [0.8842,0.9108] & 4.8867 & 0.3186 & 0.006 \\
HD1 & 0.10 & k=5 & 0.9755 & [0.9687,0.9823] & 5.3110 & 0.3712 & 0.065 \\
HD1 & 0.10 & k=10 & 0.9840 & [0.9785,0.9895] & 5.7297 & 0.2370 & 0.054 \\
HD1 & 0.10 & k=15 & 0.9915 & [0.9875,0.9955] & 6.4224 & 0.3733 & 0.055 \\
HD1 & 0.10 & mq & 0.9935 & [0.9900,0.9970] & 6.6742 & 0.3732 & 0.242 \\
HD1 & 0.10 & fisher & 0.7825 & [0.7644,0.8006] & 2.9279 & 0.1675 & 0.632 \\
HD1 & 0.10 & stouffer & 0.7565 & [0.7377,0.7753] & 2.7226 & 0.1438 & 2.350 \\
HD1 & 0.10 & simes & 0.9405 & [0.9301,0.9509] & 4.6084 & 0.5614 & 0.226 \\
HD2 & 0.05 & single & 0.9480 & [0.9383,0.9577] & 6.1337 & 0.4236 & 0.004 \\
HD2 & 0.05 & k=5 & 0.9895 & [0.9850,0.9940] & 6.7029 & 0.4654 & 0.060 \\
HD2 & 0.05 & k=10 & 0.9905 & [0.9862,0.9948] & 6.9529 & 0.3182 & 0.053 \\
HD2 & 0.05 & k=15 & 0.9950 & [0.9919,0.9981] & 7.6978 & 0.4244 & 0.054 \\
HD2 & 0.05 & mq & 0.9965 & [0.9939,0.9991] & 8.5458 & 0.5331 & 0.239 \\
HD2 & 0.05 & fisher & 0.8015 & [0.7840,0.8190] & 3.2298 & 0.1697 & 0.579 \\
HD2 & 0.05 & stouffer & 0.7730 & [0.7546,0.7914] & 2.9784 & 0.1569 & 2.319 \\
HD2 & 0.05 & simes & 0.9730 & [0.9659,0.9801] & 5.9481 & 0.6228 & 0.218 \\
HD2 & 0.10 & single & 0.8875 & [0.8737,0.9013] & 5.0596 & 0.3391 & 0.006 \\
HD2 & 0.10 & k=5 & 0.9695 & [0.9620,0.9770] & 5.6498 & 0.4169 & 0.065 \\
HD2 & 0.10 & k=10 & 0.9760 & [0.9693,0.9827] & 5.9827 & 0.2664 & 0.054 \\
HD2 & 0.10 & k=15 & 0.9855 & [0.9803,0.9907] & 6.6422 & 0.3933 & 0.055 \\
HD2 & 0.10 & mq & 0.9930 & [0.9893,0.9967] & 7.2420 & 0.4367 & 0.241 \\
HD2 & 0.10 & fisher & 0.7740 & [0.7557,0.7923] & 3.0185 & 0.1777 & 0.627 \\
HD2 & 0.10 & stouffer & 0.7405 & [0.7213,0.7597] & 2.7885 & 0.1660 & 2.357 \\
HD2 & 0.10 & simes & 0.9460 & [0.9361,0.9559] & 5.0361 & 0.5542 & 0.224

\end{longtable}
\endgroup

\subsection{Non-Gaussian dependence check}\label{app:nongaussian-dependence}
To check that the split-and-refit conclusions are not tied only to equicorrelated Gaussian dependence, we repeat the first benchmark under a non-Gaussian copula. Specifically, we draw the split dependence from a \(t\)-copula with 3 degrees of freedom and equicorrelation parameter \(\rho\in\{0.5,0.9\}\), then transform each margin to \(N(0,1)\). Thus each split-wise Gaussian pivot remains marginally calibrated, while dependence across splits is non-Gaussian and tail-dependent. Each scenario uses 20000 replications.

\begingroup
\scriptsize
\setlength{\tabcolsep}{4pt}
\begin{longtable}{cccccccc}
\caption{Non-Gaussian dependence check for the split-and-refit design. Dependence is generated by a \(t\)-copula with 3 degrees of freedom and normal margins; each scenario uses 20000 replications. Coverage CI denotes the 95\% Monte Carlo confidence interval.}\label{tab:nongaussian-dependence}\\
\toprule
\(\alpha\) & \(\rho\) & Method & Coverage & CI low & CI high & Mean len. & SD len. \\
\midrule
\endfirsthead
\toprule
\(\alpha\) & \(\rho\) & Method & Coverage & CI low & CI high & Mean len. & SD len. \\
\midrule
\endhead
\bottomrule
\endfoot
0.05 & 0.5 & Single & 0.9513 & 0.9483 & 0.9543 & 0.8624 & 0.0000 \\
0.05 & 0.5 & $k=1$ & 0.9882 & 0.9868 & 0.9897 & 0.8148 & 0.1966 \\
0.05 & 0.5 & $k=10$ & 0.9877 & 0.9861 & 0.9892 & 0.9666 & 0.0256 \\
0.05 & 0.5 & MQ & 0.9938 & 0.9927 & 0.9949 & 0.9211 & 0.1820 \\
0.05 & 0.5 & Fisher & 0.8067 & 0.8013 & 0.8122 & 0.4321 & 0.1047 \\
0.05 & 0.5 & Stouffer & 0.7630 & 0.7571 & 0.7688 & 0.4013 & 0.0844 \\
0.05 & 0.5 & Simes & 0.9778 & 0.9758 & 0.9799 & 0.7532 & 0.1759 \\
0.05 & 0.9 & Single & 0.9478 & 0.9448 & 0.9509 & 0.8624 & 0.0000 \\
0.05 & 0.9 & $k=1$ & 0.9928 & 0.9916 & 0.9940 & 1.0918 & 0.1194 \\
0.05 & 0.9 & $k=10$ & 0.9771 & 0.9750 & 0.9792 & 0.9781 & 0.0098 \\
0.05 & 0.9 & MQ & 0.9903 & 0.9890 & 0.9917 & 1.1081 & 0.0736 \\
0.05 & 0.9 & Fisher & 0.7672 & 0.7613 & 0.7731 & 0.4959 & 0.0277 \\
0.05 & 0.9 & Stouffer & 0.6814 & 0.6749 & 0.6879 & 0.4131 & 0.0185 \\
0.05 & 0.9 & Simes & 0.9714 & 0.9691 & 0.9737 & 0.9286 & 0.0686 \\
0.10 & 0.5 & Single & 0.9015 & 0.8974 & 0.9056 & 0.7237 & 0.0000 \\
0.10 & 0.5 & $k=1$ & 0.9739 & 0.9717 & 0.9762 & 0.7190 & 0.1973 \\
0.10 & 0.5 & $k=10$ & 0.9726 & 0.9703 & 0.9748 & 0.8417 & 0.0304 \\
0.10 & 0.5 & MQ & 0.9859 & 0.9843 & 0.9876 & 0.8263 & 0.1810 \\
0.10 & 0.5 & Fisher & 0.7748 & 0.7690 & 0.7805 & 0.3991 & 0.1106 \\
0.10 & 0.5 & Stouffer & 0.7254 & 0.7192 & 0.7316 & 0.3701 & 0.0911 \\
0.10 & 0.5 & Simes & 0.9510 & 0.9480 & 0.9540 & 0.6452 & 0.1724 \\
0.10 & 0.9 & Single & 0.9000 & 0.8958 & 0.9042 & 0.7237 & 0.0000 \\
0.10 & 0.9 & $k=1$ & 0.9867 & 0.9851 & 0.9882 & 0.9952 & 0.1221 \\
0.10 & 0.9 & $k=10$ & 0.9529 & 0.9500 & 0.9558 & 0.8542 & 0.0104 \\
0.10 & 0.9 & MQ & 0.9806 & 0.9786 & 0.9825 & 1.0007 & 0.0728 \\
0.10 & 0.9 & Fisher & 0.7425 & 0.7364 & 0.7486 & 0.4683 & 0.0304 \\
0.10 & 0.9 & Stouffer & 0.6534 & 0.6468 & 0.6600 & 0.3882 & 0.0209 \\
0.10 & 0.9 & Simes & 0.9421 & 0.9389 & 0.9453 & 0.8027 & 0.0663 \\

\end{longtable}
\endgroup

Table~\ref{tab:nongaussian-dependence} shows the same qualitative pattern as the Gaussian benchmark. Single split stays near the nominal level, fixed-\(k\) voting and MQ remain conservative under the non-Gaussian dependence, and Fisher/Stouffer under-cover substantially. Simes is conservative in these scenarios. These results indicate that the main empirical pattern is not an artifact of Gaussian equicorrelation.

\subsection{Real-data conformal tables (DIA and CAL)}\label{app:realdata-full}
\begingroup
\scriptsize
\setlength{\tabcolsep}{4pt}
\begin{longtable}{ccccccccc}
\caption{Real-data repeated-split conformal results on DIA (diabetes) and CAL (California housing subsample, \(n=3000\)); 3000 replications per dataset. Coverage CI denotes the 95\% Monte Carlo confidence interval.}\label{tab:realdata-full}\\
\toprule
Dataset & $\alpha$ & Tier & Method & Coverage & CI low & CI high & Mean len. & SD len. \\
\midrule
\endfirsthead
\toprule
Dataset & $\alpha$ & Tier & Method & Coverage & CI low & CI high & Mean len. & SD len. \\
\midrule
\endhead
\bottomrule
\endfoot
CAL & 0.05 & core & single & 0.9503 & 0.9426 & 0.9581 & 2.4993 & 0.0857 \\
CAL & 0.05 & core & k=5 & 0.9817 & 0.9769 & 0.9865 & 3.6224 & 0.1479 \\
CAL & 0.05 & core & k=10 & 0.9777 & 0.9724 & 0.9830 & 3.2287 & 0.0381 \\
CAL & 0.05 & core & k=15 & 0.9703 & 0.9643 & 0.9764 & 2.9686 & 0.1167 \\
CAL & 0.05 & core & MQ & 0.9860 & 0.9818 & 0.9902 & 3.8446 & 0.1032 \\
CAL & 0.05 & appendix & fisher & 0.7590 & 0.7437 & 0.7743 & 1.1628 & 0.0705 \\
CAL & 0.05 & appendix & stouffer & 0.6617 & 0.6447 & 0.6786 & 0.9328 & 0.0580 \\
CAL & 0.05 & appendix & simes & 0.9623 & 0.9555 & 0.9691 & 2.7849 & 0.1517 \\
CAL & 0.10 & core & single & 0.8990 & 0.8882 & 0.9098 & 1.8051 & 0.0437 \\
CAL & 0.10 & core & k=5 & 0.9713 & 0.9654 & 0.9773 & 3.0915 & 0.1496 \\
CAL & 0.10 & core & k=10 & 0.9530 & 0.9454 & 0.9606 & 2.4782 & 0.0414 \\
CAL & 0.10 & core & k=15 & 0.9360 & 0.9272 & 0.9448 & 2.1802 & 0.1091 \\
CAL & 0.10 & core & MQ & 0.9773 & 0.9720 & 0.9827 & 3.2529 & 0.1304 \\
CAL & 0.10 & appendix & fisher & 0.7327 & 0.7168 & 0.7485 & 1.0968 & 0.0713 \\
CAL & 0.10 & appendix & stouffer & 0.6263 & 0.6090 & 0.6436 & 0.8749 & 0.0565 \\
CAL & 0.10 & appendix & simes & 0.9233 & 0.9138 & 0.9329 & 2.0503 & 0.1312 \\
DIA & 0.05 & core & single & 0.9500 & 0.9422 & 0.9578 & 2.9115 & 0.1378 \\
DIA & 0.05 & core & k=5 & 0.9790 & 0.9739 & 0.9841 & 3.3162 & 0.1073 \\
DIA & 0.05 & core & k=10 & 0.9750 & 0.9694 & 0.9806 & 3.2568 & 0.0613 \\
DIA & 0.05 & core & k=15 & 0.9790 & 0.9739 & 0.9841 & 3.3132 & 0.0963 \\
DIA & 0.05 & core & MQ & 0.9910 & 0.9876 & 0.9944 & 3.8367 & 0.0789 \\
DIA & 0.05 & appendix & fisher & 0.7607 & 0.7454 & 0.7759 & 1.7110 & 0.0301 \\
DIA & 0.05 & appendix & stouffer & 0.6677 & 0.6508 & 0.6845 & 1.4226 & 0.0239 \\
DIA & 0.05 & appendix & simes & 0.9700 & 0.9639 & 0.9761 & 3.1675 & 0.0811 \\
DIA & 0.10 & core & single & 0.9010 & 0.8903 & 0.9117 & 2.4521 & 0.0963 \\
DIA & 0.10 & core & k=5 & 0.9603 & 0.9533 & 0.9673 & 2.9720 & 0.1065 \\
DIA & 0.10 & core & k=10 & 0.9560 & 0.9487 & 0.9633 & 2.8640 & 0.0509 \\
DIA & 0.10 & core & k=15 & 0.9590 & 0.9519 & 0.9661 & 2.8905 & 0.0928 \\
DIA & 0.10 & core & MQ & 0.9830 & 0.9784 & 0.9876 & 3.4573 & 0.0657 \\
DIA & 0.10 & appendix & fisher & 0.7343 & 0.7185 & 0.7501 & 1.6218 & 0.0308 \\
DIA & 0.10 & appendix & stouffer & 0.6423 & 0.6252 & 0.6595 & 1.3417 & 0.0244 \\
DIA & 0.10 & appendix & simes & 0.9493 & 0.9415 & 0.9572 & 2.7582 & 0.0894 \\

\end{longtable}
\endgroup

\subsection{Real-data conformal runtime benchmark}\label{app:realdata-runtime}
\begingroup
\scriptsize
\setlength{\tabcolsep}{4pt}
\begin{longtable}{cccc}
\caption{Runtime benchmark for real-data repeated-split conformal experiments at $\alpha=0.05$ (1000 replications per dataset). Values are milliseconds per replication.}\label{tab:realdata-runtime}\\
\toprule
Method & DIA & CAL & Mean \\
\midrule
\endfirsthead
\toprule
Method & DIA & CAL & Mean \\
\midrule
\endhead
\bottomrule
\endfoot
single & 0.006 & 0.007 & 0.006 \\
$k=5$ & 0.034 & 0.035 & 0.034 \\
$k=10$ & 0.029 & 0.029 & 0.029 \\
$k=15$ & 0.030 & 0.030 & 0.030 \\
MQ & 0.125 & 0.135 & 0.130 \\
Fisher & 0.359 & 0.366 & 0.362 \\
Stouffer & 1.516 & 1.553 & 1.535 \\
Simes & 0.144 & 0.144 & 0.144 \\
\end{longtable}
\endgroup

\subsection{Conformal oracle summary}
\begingroup
\scriptsize
\setlength{\tabcolsep}{2pt}
\begin{longtable}{@{}ccccccccc@{}}
\caption{Oracle-$k$ benchmark among fixed-$k$ rules and relative MQ efficiency for repeated-split conformal ensembling.}\label{tab:conformal-oracle}\\
\toprule
Scenario & $n$ & $\alpha$ & Oracle $k^\star$ & Oracle cov. & Oracle len. & MQ cov. & MQ len. & MQ/oracle \\
\midrule
\endfirsthead
\toprule
Scenario & $n$ & $\alpha$ & Oracle $k^\star$ & Oracle cov. & Oracle len. & MQ cov. & MQ len. & MQ/oracle \\
\midrule
\endhead
\bottomrule
\endfoot
S1 & 600 & 0.05 & 17 & 0.9701 & 2.5904 & 0.9907 & 3.1756 & 1.2259 \\
S1 & 600 & 0.10 & 18 & 0.9343 & 2.2089 & 0.9801 & 2.7980 & 1.2667 \\
S2 & 600 & 0.05 & 17 & 0.9692 & 3.6352 & 0.9915 & 4.6833 & 1.2883 \\
S2 & 600 & 0.10 & 18 & 0.9365 & 3.0178 & 0.9792 & 3.9907 & 1.3224 \\
S1 & 900 & 0.05 & 19 & 0.9633 & 2.5544 & 0.9905 & 3.1069 & 1.2163 \\
S1 & 900 & 0.10 & 19 & 0.9271 & 2.1679 & 0.9761 & 2.7490 & 1.2681 \\
S2 & 900 & 0.05 & 19 & 0.9657 & 3.5713 & 0.9885 & 4.5522 & 1.2747 \\
S2 & 900 & 0.10 & 19 & 0.9295 & 2.9549 & 0.9760 & 3.9040 & 1.3212 \\
\end{longtable}
\endgroup

\subsection{Multidimensional runtime stress test summary}\label{app:multidim-runtime-full}
\begingroup
\scriptsize
\setlength{\tabcolsep}{4pt}
\begin{longtable}{ccccccc}
\caption{Detailed multidimensional runtime summaries by dimension and method. Mean/median/SD are milliseconds per replication over 20 benchmark replications.}\label{tab:multidim-runtime-full}\\
\toprule
$d$ & $G$ & $G^d$ & Method & Mean ms & Median ms & SD ms \\
\midrule
\endfirsthead
\toprule
$d$ & $G$ & $G^d$ & Method & Mean ms & Median ms & SD ms \\
\midrule
\endhead
\bottomrule
\endfoot
1 & 15 & 15 & k=1-vote & 0.088 & 0.081 & 0.026 \\
1 & 15 & 15 & fisher-grid & 0.265 & 0.252 & 0.045 \\
1 & 15 & 15 & stouffer-grid & 0.384 & 0.379 & 0.020 \\
1 & 15 & 15 & simes-grid & 0.156 & 0.151 & 0.014 \\
2 & 15 & 225 & k=1-vote & 0.090 & 0.087 & 0.015 \\
2 & 15 & 225 & fisher-grid & 0.683 & 0.670 & 0.042 \\
2 & 15 & 225 & stouffer-grid & 1.099 & 1.082 & 0.047 \\
2 & 15 & 225 & simes-grid & 0.540 & 0.538 & 0.006 \\
3 & 15 & 3375 & k=1-vote & 0.124 & 0.120 & 0.018 \\
3 & 15 & 3375 & fisher-grid & 8.634 & 8.647 & 0.245 \\
3 & 15 & 3375 & stouffer-grid & 14.487 & 14.473 & 0.316 \\
3 & 15 & 3375 & simes-grid & 8.442 & 8.520 & 0.300 \\
4 & 15 & 50625 & k=1-vote & 0.184 & 0.182 & 0.022 \\
4 & 15 & 50625 & fisher-grid & 155.135 & 153.458 & 6.924 \\
4 & 15 & 50625 & stouffer-grid & 250.463 & 248.520 & 5.501 \\
4 & 15 & 50625 & simes-grid & 152.565 & 152.335 & 1.488 \\
5 & 15 & 759375 & k=1-vote & 0.199 & 0.181 & 0.043 \\
5 & 15 & 759375 & fisher-grid & 2493.705 & 2494.783 & 20.289 \\
5 & 15 & 759375 & stouffer-grid & 3906.422 & 3901.805 & 28.900 \\
5 & 15 & 759375 & simes-grid & 2489.577 & 2486.195 & 19.071 \\

\end{longtable}
\endgroup

\newpage

\bibliographystyle{agsm}
\bibliography{refs}

\end{document}